\documentclass[review,hidelinks,onefignum,onetabnum]{siamart220329}
\usepackage{Basic_siam}

\usepackage{makecell}
\usepackage{tikz}
\usetikzlibrary{positioning, shapes.geometric, arrows}
\usepackage{booktabs}
\usepackage{subfigure}

\title{Estimate of Koopman modes and eigenvalues with Kalman Filter\thanks{Submitted to the editors DATE.}}

\author{Ningxin Liu\footnotemark[4] \thanks{These authors contributed to the work equally and should be regarded as co-first authors.}
\and Shuigen Liu\footnotemark[2] \thanks{National University of Singapore (\email{shuigen@u.nus.edu}, \email{mattxin@nus.edu.sg})} 
\and Xin T. Tong\footnotemark[3] 
\and Lijian Jiang\thanks{Tongji University (\email{nxliu@tongji.edu.cn}, \email{ljjiang@tongji.edu.cn})}}

\headers{Estimate of Koopman modes and eigenvalues with Kalman Filter}{Ningxin Liu, Shuigen Liu, Xin T. Tong, Lijian Jiang}

\ifpdf
\hypersetup{
pdftitle={Estimate of Koopman modes and eigenvalues with Kalman Filter}{Ningxin Liu \and Shuigen Liu \and Xin T. Tong \and Lijian Jiang}
}
\fi

\newcommand{\bphi}{{\bm{\phi}}}
\newcommand{\bPhi}{{\bm{\Phi}}}
\newcommand{\bLambda}{{\bm{\Lambda}}}
\newcommand{\beps}{{\bm{\epsilon}}}
\newcommand{\bb}{{\bm{b}}}
\newcommand{\bbf}{{\bm{f}}}
\newcommand{\bvareps}{{\bm{\varepsilon}}}
\newcommand{\btheta}{{\bm{\theta}}}
\newcommand{\bTheta}{{\bm{\Theta}}}

\newcommand{\bxi}{{\bm{\xi}}}
\newcommand{\by}{{\bm{y}}}
\newcommand{\bh}{{\bm{h}}}
\newcommand{\bY}{{\bm{Y}}}

\newcommand{\bg}{{\bm{g}}}

\newcommand{\bK}{{\bm{K}}}

\begin{document}
\maketitle

\begin{abstract}
Dynamic mode decomposition (DMD) is a data-driven method of extracting spatial-temporal coherent modes from complex systems and providing an equation-free architecture to model and predict systems. However, in practical applications, the accuracy of DMD can be limited in extracting dynamical features due to sensor noise in measurements. We develop an adaptive method to constantly update dynamic modes and eigenvalues from noisy measurements arising from discrete systems. Our method is based on the Ensemble Kalman filter owing to its capability of handling time-varying systems and nonlinear observables. Our method can be extended to non-autonomous dynamical systems, accurately recovering short-time eigenvalue-eigenvector pairs and observables. Theoretical analysis shows that the estimation is accurate in long term data misfit. We demonstrate the method on both autonomous and non-autonomous dynamical systems to show its effectiveness.
\end{abstract}

\begin{keywords}
Dynamic mode decomposition, Koopman modes, ensemble Kalman filter, non-autonomous system
\end{keywords}

\begin{MSCcodes}
37M99 47D99 65C20
\end{MSCcodes}

\section{Introduction}
In recent years, data-driven methods have become important computational tools in applied science and engineering. With the rapid advancement in data collection, a wide range of computational methods have been developed to leverage data to improve the modeling and prediction of complex dynamical systems. These approaches connect observables and underlying patterns without explicit knowledge of the physical behavior of the system. One effective strategy involves decomposing the dynamics into multiple components to investigate complex phenomena through decomposed modes. For example, extracting relevant features from large-scale measurements and characterizing the spatiotemporal coherent structures make fluids more tractable for engineering design, analysis, and external control \cite{proctor2016dynamic}, and help explain complex physical phenomena through individually decomposed components. Proper orthogonal decomposition (POD) \cite{berkooz1993proper} is one of the earliest data-driven modal decomposition techniques in fluid dynamics,  decomposing the dynamics into orthogonal modes that optimally span the data.

Dynamic mode decomposition (DMD) \cite{rowley2009spectral,schmid2010dynamic,kutz2016dynamic} is a data-driven framework designed to extract spatial-temporal coherent modes and reconstruct the underlying dynamics of complex systems. The method does not require knowledge of the underlying governing equations; it only needs snapshots of the dynamical system. The Koopman operator \cite{mezic2013analysis,budivsic2012applied} is an infinite-dimensional linear operator that characterizes the behavior of nonlinear systems with specific observables. There is a direct connection between DMD methods and Koopman operator theory \cite{rowley2009spectral,mezic2013analysis}. DMD approximates the eigenvalues and eigenfunctions of the Koopman operator for identification and diagnostics of features in complex dynamics using simulated or experimental data. For example, DMD has been extensively used in fluid dynamics \cite{schmid2011application,jardin2012lagrangian,li2021deep} to investigate a wide range of flow phenomena, which often involve large-scale datasets. The methods are also widely used in fields such as image processing \cite{kutz2016multiresolution}, video processing \cite{bouwmans2016handbook}, financial trading \cite{mann2016dynamic} and medical care \cite{bourantas2014real}, and forecasting the spread of infectious diseases \cite{proctor2015discovering}.

One major challenge to the DMD method is the presence of noise within snapshots. Noisy snapshots result in inaccurate DMD eigenvalues and modes, leading to errors in the reconstruction and prediction of the observables. Various efforts have been made to address this issue. Dawson and Hemati \cite{dawson2016characterizing} introduced a method called Noise-corrected DMD, which directly corrects sensor noise. Forward-backward DMD \cite{bouwmans2016handbook} and Total-least-square DMD \cite{hemati2017biasing} propose to approximate the DMD operator by considering new types of approximate DMD matrix. Kalman filter-DMD \cite{jiang2022correcting} improves the accuracy of approximating the DMD operator by filtering noisy observables. Another method, proposed by N. Takeishi et al. \cite{takeishi2017bayesian}, estimates the spectrum of the Koopman operator from a Bayesian and sampling perspective. Additionally, Nonomura et al. \cite{nonomura2018dynamic} proposed a parameter estimation method for the Koopman matrix using the Kalman filter, treating the Koopman matrix as the state variable and estimating it directly.

The aforementioned numerical methods for Koopman operator estimation often assume the system to be autonomous. However, many problems, such as two-parameter processes and the skew product flows, are non-autonomous. Kutz et al. \cite{kutz2016multiresolution} proposed multiresolution dynamic mode decomposition (mrDMD), which separates complex systems into a hierarchy of multiresolution time-scale components. The mrDMD successfully decomposes the multiscale data generated from an underlying time-dependent linear system. In \cite{mezic2016koopman}, the Koopman operator framework was further extended to the non-autonomous systems with periodic or quasi-periodic time dependence, rigorously defining the time-dependent eigenfunctions, eigenvalues, and modes of the non-autonomous Koopman operator. In recent years, more numerical methods for time-dependent Koopman spectral approximation have emerged, such as extensions of DMD and Extended DMD. In \cite{giannakis2019data}, Giannakis developed a technique that involves rescaling the generator in the time-changed setting and applied the approach to mixing dynamical systems where the generator has no constant eigenfunctions. This approach, incorporating delay coordinates was further applied in \cite{das2019delay} to approximate Koopman eigenfunctions systems with pure point or mixed spectra. An approach for learning the time-dependent Koopman operator in multiscale media was presented in \cite{li2023data}, where the authors used a moving time window to localize snapshots. Other studies, like \cite{proctor2016dynamic,proctor2016including,korda2018linear}, also generalize the Koopman operator to non-autonomous systems, but focusing on input and control. However, these existing works are carried out for deterministic systems without noise or stochasticity, which can be too idealistic for some applications.

In this article, we propose an Ensemble Kalman filter (EnKF) based DMD method, which can estimate DMD modes and eigenvalues for both autonomous and non-autonomous systems. In particular, we identify DMD modes and eigenvalues in a Bayesian framework and estimate them with an EnKF. The method is suitable for temporal sequential data since it achieves online estimation of eigenvector-eigenvalue pairs when new data is added.  With the advantage of EnKF, we can obtain accurate estimations from noisy snapshots of known noise magnitude. We use time-delay snapshots as observable functions in EnKF. The origin of this method can be found in the moving stencils approach for non-autonomous systems \cite{macesic2018koopman}. The method is also applied to non-autonomous dynamical systems to estimate time-dependent Koopman eigenvalues and eigenfunctions. In our application, the moving stencil of snapshots is used to approximate the Koopman operator family. 

The paper is organized as follows . In \Cref{sec2}, we give a preliminary introduction to the Koopman operator and DMD method. In \Cref{sec3}, we propose the EnKF-DMD method to update Koopman modes and eigenvalues given noisy snapshots, and show how the algorithm can be applied to both autonomous systems and non-autonomous systems with minor adjustments. In \Cref{sec4}, we prove that EnKF-DMD provides consistent predictions, with the long term error between the reconstructed observable and the truth bounded at the scale of observation noise and interval. In \Cref{sec5}, we implement EnKF-DMD on various dynamical systems, which confirms the effectiveness and accuracy of our method.

\section{Preliminaries}
\label{sec2}
\subsection{Koopman operator for autonomous dynamical systems and DMD}
We consider a dynamical system
\begin{equation}    \label{c-system}
\diff{x}{t}=f(x), \quad x(t)\in\mathcal{M}
\end{equation}
where $\mathcal{M}$ is a smooth  manifold and $f:\mathcal{M}\rightarrow T \mathcal{M}$ is a nonlinear map. Here $T \mathcal{M}$ is the tangent bundle. We sample the data with an interval $\Delta t$ time and denote the subscript as the time index so that $x_k=x(k\Delta t)$. Then the discrete-time system representation corresponding to the continuous-time dynamical system \cref{c-system} is denoted as
\begin{equation}
\label{d-system}
x_{k+1}=F(x_k),\quad x_k\in\mathcal{M}.
\end{equation}
Let $g:\mathcal{M}\rightarrow\mathbb{C}$ be a scalar observable function in a function space $\mathcal{G}$. The discrete-time 
Koopman operator $\mathcal{K}:\mathcal{G}\rightarrow\mathcal{G}$ is an infinite-dimensional linear operator that acts on all observable functions $g$ so that
\[
\mathcal{K}g(x) = g(F(x)),\quad g\in \mathcal{G}.
\] 
The Koopman operator induces a discrete-time dynamical system on the observable $g$:
\begin{equation}
\label{evolu}
\mathcal{K}g(x_k)=g(F(x_k))=g(x_{k+1}).
\end{equation}
Now we consider the spectral decomposition of the Koopman operator to represent the evolution of the dynamical systems of interest. Define $\varphi_i:\mathcal{M}\rightarrow\mathbb{C}$ as eigenfunctions of $\mathcal{K}$ with eigenvalues $\lambda_i\in\mathbb{C}$, i.e.,
\[\mathcal{K}\varphi_i = \lambda_i\varphi_i.\]
The functions $\varphi_i$ are a set of coordinates to represent observable functions in $\mathcal{G}$. These coordinates can be used to represent the advance of observable functions in the linear dynamical system \cref{evolu}. Then the evolution of the dynamical systems based on observables is expressed as an expansion of eigenfunctions of the Koopman operator.

Let $\bg:\mathcal{M}\rightarrow\mathbb{R}^d$ be a vector of observables written in terms of Koopman eigenfunctions $\varphi_k$ as
\begin{equation}
\label{observabl_1}
\bg(x_k)= \begin{bmatrix}
g_1(x_k),
g_2(x_k), \ldots, 
g_d(x_k)
\end{bmatrix}\matT =\sum_{i=1}^{\infty}\varphi_i(x_k)\bm{\phi}_i,
\end{equation}
where $\bm{\phi}_i\in \mathbb{R}^d$ is the $i$-th Koopman mode associated with the $i$-th Koopman eigenfunction $\varphi_i$. The Koopman operator acts on both sides of \cref{observabl_1} to get
\begin{equation}
\label{observabl_2}
\bg(x_{k+1})=\mathcal{K}\bg(x_{k})=\sum_{i=1}^{\infty}\lambda_i\varphi_i(x_k)\bm{\phi}_i.
\end{equation}

In this way, the original nonlinear system \cref{d-system} is represented by a linear evolution of observable functions. If a $r$-dimensional truncation is applied, the modal decomposition of $\bg(x_k)$ is given by 
\begin{equation}
\label{decomposition}
\bg(x_k)=\sum_{i=1}^{r}\lambda^k_ib_i\bm{\phi}_i,\quad b_i=\varphi_i(x_0).
\end{equation}
Dynamical mode decomposition (DMD) \cite{kutz2016dynamic,schmid2010dynamic} is a method used to approximate the Koopman eigenvalues $\lambda_i$ and modes $\bm{\phi}_i$. Suppose that we have two data matrices $\bY_0\in \mathbb{R}^{d\times m}$  and $\bY_1\in \mathbb{R}^{d\times m}$ whose columns denote the observables $\bg$:
\begin{align}
\label{data-matrices}
\bY_0 &= \begin{bmatrix}
\bg(x_0), & \bg(x_1), & \ldots, & \bg(x_{m-1})
\end{bmatrix}, \\
\bY_1 &= \begin{bmatrix}
\bg(x_1), & \bg(x_2), & \ldots, & \bg(x_m)
\end{bmatrix}.
\end{align}
DMD produces a matrix $\widehat{\bK}$ fitting the
observable trajectory in the sense of  least-squares  such that
\begin{equation}
\label{minimum}
\widehat{\bK} = \argmin_{\bK}\norm{ \bY_1-\bK\bY_0}_{\rm F}.
\end{equation}
The method approximates a finite-dimensional operator 
$\widehat{\bK}=\bY_1\bY_0^{\dagger}$
that maps the columns of $\bY_0$ to $\bY_1$. 
The corresponding algorithm, Exact DMD \cite{tu2013dynamic},  is formalized in \cref{alg:1}. 

\begin{algorithm}[htb]
	\caption{Exact DMD}
	\label{alg:1}
	\begin{algorithmic}[1]
	\STATE Compute the compact SVD of 
        $\bY_0=\bm{U}_r\mathbf{\Sigma}_r\bm{V}_r\matH$
        \STATE Define a matrix $\widetilde{\bK}=\bm{U}_r\matH\bY_1\bm{V}_r\bm{\Sigma}_r^{-1}$
        \STATE Calculate nonzero eigenvalues and eigenvectors of $\widetilde{\bK}$,
        i.e., compute $\widetilde{\bm{\phi}}$ and $\mathbf{\lambda}$ such that $\widetilde{\bK}\widetilde{\bm{\phi}}=\mathbf{\lambda}\widetilde{\bm{\phi}}$
        \STATE The DMD mode $\bm{\phi}$ is given by
$\bm{\phi}=\bY_1\bm{V}_r\mathbf{\Sigma}_r^{-1}\widetilde{\bm{\phi}}$ corresponding to eigenvalue $\mathbf{\lambda}$
	\end{algorithmic}
\end{algorithm}

\subsection{Koopman operator for non-autonomous dynamical systems}
Let us consider the following nonlinear non-autonomous dynamical system
\begin{equation}
\label{non-sys}
\left\{
\begin{aligned}
&\diff{x}{t}=f(x,t), \;x\in\mathcal{M},  \\
&x(0)=x_0.
\end{aligned}
\right.
\end{equation}
We define a mapping $\mathcal{S}:\mR_+ \times \mR_+ \times \mathcal{M} \rightarrow \mathcal{M}$ such that this two-parameter family   $\mathcal{S}^{t,t_0}=\mathcal{S}(t,t_0,\cdot)$ satisfying the cocycle property, i.e.,
\begin{equation}
\label{cocycle_S}
\mathcal{S}^{u,t}\circ\mathcal{S}^{t,s}=\mathcal{S}^{u,s},\quad
\mathcal{S}^{t_0,t_0}=\mathcal{I}.  
\end{equation}
Suppose that $\mathcal{S}^{t,t_0}$ is generated by the non-autonomous dynamical system \cref{non-sys}, then the solution of this system can be written in the form of
\[x(t)=\mathcal{S}^{t,t_0}x_0.\]
Let $\bg(x):\mathcal{M}\rightarrow\mathbb{R}^d$ still be an observable vector in space $\mathcal{F}$. Now we define the two-parameter Koopman operator family $\mathcal{K}^{t,t_0}$ \cite{schmid2010dynamic}: 
\[
\mathcal{K}^{t,t_0}\bg = \bg\circ\mathcal{S}^{t,t_0}
\]
Let $\{\lambda_i^{t,t_0},\varphi_i^{t,t_0}\}_{i=1}^r$ be the  eigenvalues and eigenfunctions of  $\mathcal{K}^{t,t_0}$, i.e.,
\[
\mathcal{K}^{t,t_0}\varphi_{i}^{t,t_0}(x) = {\lambda_i^{t,t_0}}\varphi_{i}^{t,t_0}(x).
\]
Then related decomposition of observables via the non-autonomous Koopman operator is given by:
\begin{equation}
\label{decomposition-non}
\bg(x_t)=\mathcal{K}^{t,t_0}\bg(x_{0})=\sum_{i=1}^{d} {\lambda^{t,t_0}_i} \varphi_i^{t,t_0}(x_0)\bm{\phi}_i^{t,t_0},
\end{equation}
where $\bm{\phi}_i^{t,t_0}$ is the $i$-th Koopman mode associated with the $i$-th Koopman eigenvalue $\lambda^{t,t_0}_i$. In non-autonomous dynamical systems, the modes and eigenvalues of the Koopman operator change with time. With the cocycle property \cref{cocycle_S}, we can derive that non-autonomous Koopman operator also satisfies the cocycle property such that
\begin{equation}
\label{cocycle_K}
\mathcal{K}^{u,t}\circ\mathcal{K}^{t,s}=\mathcal{K}^{u,s},\quad
\mathcal{K}^{t_0,t_0}=\mathcal{I}.
\end{equation}

\section{Koopman Spectral estimation with EnKF} \label{sec3}
It is known that the exact DMD method (\cref{alg:1}) is sensitive to noise in the dataset. To improve noise robustness, we propose using EnKF \cite{evensen2003ensemble} to estimate the Koopman modes and eigenvalues, resulting in a Bayesian DMD method. EnKF is a Bayesian mechanism tracking the dynamic of a system given noisy observations. To apply EnKF, we treat the eigenvalues and modes as states in some dynamical system and use time-delay snapshots as observables to correct the estimation. We will show how this method can be applied to both autonomous and non-autonomous dynamical systems.

\subsection{Review on Ensemble Kalman Filter} 
In this section, we give a brief review on filtering and EnKF before introducing our method. In filtering, one considers a dynamical system with state $\btheta_k$ and its noisy observation $\by_k$: 
\begin{subequations}
\begin{align} 
    \btheta_{k} =~& \bbf_k (\btheta_{k-1}) + \bxi_k,  \label{eq:EnKF_evol}  \\
    \by_k =~& \bh_k(\btheta_k) + \bvareps_k.  \label{eq:EnKF_obs}
\end{align}   
\end{subequations}
Here $\bbf_k$ is the evolution of the state, $\bh_k$ is the observation function, and $\bxi_k$ and $\bvareps_k$ are the state noise and observation noise respectively. 

The objective of filtering  is to estimate the posterior distribution of $\btheta_k$ given the observables $\by_{1:k} := (\by_1,\by_2,\cdots,\by_k)$. The posterior distribution is given by the recursive Bayesian formula:
\[
    \mP(\btheta_k|\by_{1:k}) = \frac{ \mP(\by_k|\btheta_k, \by_{1:k-1}) \mP(\btheta_k|\by_{1:y-1}) }{ \mP( \by_k | \by_{1:y-1}) } = \frac{ \mP(\by_k|\btheta_k) \mP(\btheta_k|\by_{1:k-1}) }{ \mP( \by_k | \by_{1:k-1}) }.
\]
Since from \cref{eq:EnKF_obs}, $\by_k$ is conditionally independent of $\by_{1:k-1}$ given $\btheta_k$, we have
\begin{equation}    \label{eq:exact_anal}
    \mP(\btheta_k|\by_{1:k}) \propto \mP(\by_k|\btheta_k) \mP(\btheta_k|\by_{1:k-1}).
\end{equation}
Here $\mP(\by_k|\btheta_k)$ is the likelihood for $\by_k$ and is computable. To compute $\mP(\btheta_k|\by_{1:k-1})$, notice from \cref{eq:EnKF_evol}, $\btheta_k$ is conditionally independent of $\by_{1:k-1}$ given $\btheta_{k-1}$, and thus
\begin{equation}    \label{eq:exact_pred}
    \mP(\btheta_k|\by_{1:k-1}) = \int \mP(\btheta_k|\btheta_{k-1}) \mP(\btheta_{k-1}|\by_{1:k-1}) \mdd \btheta_{k-1}.
\end{equation}
\cref{eq:exact_anal} and \cref{eq:exact_pred} provide exact solutions to update the target distribution $\mP(\btheta_k|\by_{1:k})$ in a recursive manner. However, the above exact quadrature solution is computationally infeasible in practice, and one may seek practical approximations, among which a popular choice is the EnKF \cite{evensen2009data}. 

In EnKF, the posterior distribution $\mP(\btheta_k|\by_{1:k})$ is approximated by a Gaussian distribution estimated using an ensemble $\{\btheta_k^{(i)}\}_{i=1}^N$. For simplicity of presentation, consider the case where the state $\btheta_k$ is invariant, i.e. $\bbf_k(\btheta) \equiv \btheta$ in  \cref{eq:auto_evol}. Under this setting, 
EnKF takes a simpler form that approximates \cref{eq:exact_pred} and \cref{eq:exact_anal} by 
\begin{equation}    \label{eq:EnKF}
    \btheta_{k+1}^{(i)} = \btheta_k^{(i)} + K_k \Brac{ \widehat{\by}_k - \bh_k(\btheta_k^{(i)}) - \bvareps_k^{(i)} } + \bxi_k^{(i)},
\end{equation}
\begin{itemize}
    \item $K_k$ is known as the \textit{Kalman gain} defined by
\begin{equation}    \label{eq:Kal_gain}
    K_k := P_k^{\btheta,\by} \Brac{ P_k^{\by,\by}+ \sigma^2 I }^{-1},
\end{equation}
where we define the empirical mean and covariance matrices  
\begin{equation}   \label{eq:cov} 
\begin{split}
    &\mean{\btheta}_k := \frac{1}{N} \sum_{i=1}^N \btheta_k^{(i)}, \quad \mean{\bh_k(\btheta_k)} := \frac{1}{N} \sum_{i=1}^N \bh_k(\btheta_k^{(i)}), \\
    &P_k := \frac{1}{N-1} \sum_{i=1}^N \Brac{ \btheta_k^{(i)} - \mean{\btheta}_k } \otimes \Brac{ \btheta_k^{(i)} - \mean{\btheta}_k }, \\
    &P_k^{\btheta,\by} := \frac{1}{N-1} \sum_{i=1}^N \Brac{ \btheta_k^{(i)} - \mean{\btheta}_k } \otimes \Brac{ \bh_k(\btheta_k^{(i)}) - \mean{\bh_k(\btheta_k)} }, \\
    &P_k^{\by,\by} := \frac{1}{N-1} \sum_{i=1}^N \Brac{ \bh_k(\btheta_k^{(i)}) - \mean{\bh_k(\btheta_k)} } \otimes \Brac{ \bh_k(\btheta_k^{(i)}) - \mean{\bh_k(\btheta_k)} }.
\end{split}
\end{equation}
$P_k^{\by,\btheta}$ is defined similarly and note $ P_k^{\by,\btheta} = (P_k^{\btheta,\by})\matT $.

\item $\widehat{\by}_k$ is one realization of the noised observation $\by_k$, i.e. the data one obtained.  

\item $\bxi_k^{(i)},\bvareps_k^{(i)}$ are artificial noises usually drawn from Gaussian distributions
\[
    \bxi_k^{(i)} \topsm{\iid}{\sim} \mcN(0,Q), \quad \bvareps_k^{(i)} \topsm{\iid}{\sim} \mcN(0,\sigma^2 I), \quad i =1,\dots, N,
\]
and $Q,\sigma$ are chosen to match the covariance matrices of the actual noise. Here we assume the observation noise is isotropic for simplicity, and note one can always normalize $\by$ to make it isotropic. 
\end{itemize}
In expectation, the mean and covariance of \cref{eq:EnKF} would follow 
\begin{subequations}
\begin{align}
    \mean{\btheta}_{k+1} &= \mean{\btheta}_k + K_k \Brac{ \widehat{\by}_k - \mean{\bh_k(\btheta_k)} }, \label{eq:ETKF_mean} \\
    P_{k+1} =~& P_k - P_k^{\btheta,\by} \Brac{ P_k^{\by,\by}+ \sigma^2 I }^{-1} P_k^{\by,\btheta} + Q,   \label{eq:ETKF_cov}
\end{align}
\end{subequations} 
For vanilla EnKF, the mean and covariance update \cref{eq:ETKF_mean}, \cref{eq:ETKF_cov} are only satisfied in expectation. To ensure the exact update, we will use a modified scheme known as the ensemble transform Kalman filter (ETKF). In essence, ETKF aims to find the transform operator $T_k$ to ensure
\[
    \frac{ \Brac{T_k \bTheta_k} \otimes \Brac{T_k \bTheta_k}}{N-1} = P_{k+1} = P_k - P_k^{\btheta,\by} \Brac{ P_k^{\by,\by}+ \sigma^2 I }^{-1} P_k^{\by,\btheta} + Q, 
\]
where we denote $\bTheta_k$ as the deviation matrix of $\btheta_k^{(i)}$:  
\[
    \bTheta_k = \Rectbrac{ \btheta_k^{(1)} - \mean{\btheta}_k , \cdots, \btheta_k^{(N)} - \mean{\btheta}_k }.
\]
With the transform operator $T_k$, the covariance update is taken as $\bTheta_{k+1} = T_k \bTheta_k$. In this way, the update \cref{eq:ETKF_mean} and \cref{eq:ETKF_cov} are strictly satisfied in ETKF. For implementation details, interested readers are referred to \cite{bishop2001adaptive}. 

\subsection{Autonomous dynamical systems}
For autonomous discrete systems, the spectrum of the Koopman operator is time-independent. This allows us to treat it as a parameter vector that remains constant over time. To apply EnKF to estimate Koopman modes and eigenvalues, we need to specify the dynamic model \cref{eq:EnKF_evol}, \cref{eq:EnKF_obs}. We introduce our method as follows. 

Based on the spectral decomposition \cref{decomposition} of the Koopman operator, it suffices to estimate the triple of the Koopman modes, eigenvalues and the initial conditions $\{\phi_i,\lambda_i,b_i\}_{i=1}^r$. Denote these parameters as
\begin{equation}    \label{eq:para}
    \btheta = \begin{bmatrix}
    \bphi\matT_1, \cdots, \bphi\matT_r, \lambda_1, \cdots, \lambda_r, b_1, \cdots, b_r
    \end{bmatrix}\matT.
\end{equation}
We will treat $\btheta$ as the state in the dynamical system \cref{eq:EnKF_evol}. As the Koopman operator \cref{evolu} is time-independent for automonous systems, $\btheta$ is also time-independent. Therefore, the dynamics of $\btheta$ can be simply taken as
\begin{equation}    \label{eq:auto_evol}
    \btheta_{k+1} = \btheta_k. 
\end{equation}
Next we determine the observation model \cref{eq:EnKF_obs}, i.e. the function $\bh_k(\btheta_k)$. Notice by the decomposition \cref{decomposition}, the observation function $\bg(x_k)$ can be written as a function of $\bm{\theta}$: 
\begin{equation}    \label{eq:hk}
    \bg(x_k) = \sum_{i=1}^r \lambda_i^k b_i \phi_i =: h_k(\btheta). 
\end{equation}
It can be rewritten in a matrix form $h_k(\btheta) = \bPhi \bLambda^k \bb$, where
\[
    \bPhi = \begin{bmatrix}
        \Phi_1, \dots, \Phi_r
    \end{bmatrix}, \quad \bLambda = \text{diag}(\lambda_1,\cdots,\lambda_r), \quad \bb = \begin{bmatrix}
        b_1, \cdots, b_r
    \end{bmatrix}\matT.
\]

One can aleady use $h_k$ as the observation function at time $k$ to run EnKF for $\btheta$. But this is not a good choice. Using observation at one single time is often insufficient to capture the dynamical characteristics. A better choice is to use time-delay snapshots of the observation function, which can be used for accurate reconstruction due to the  delay embedding theorems. More specifically, Taken's embedding theorem \cite{takens2006detecting} shows that one can indeed reconstruct qualitative features of the dynamic in $\mR^d$ with $N>2d$ number of time-delay measurements. This motivates us to replace $h_k$ by its time-delay version. To be specific, we take 
\begin{equation}    \label{eq:obs}
\begin{split}
    \bh_k (\btheta) := \begin{bmatrix}
        h_{k-n}(\btheta) \\
        h_{k-n+1}(\btheta) \\
        \vdots \\
        h_k(\btheta)
    \end{bmatrix} = \begin{bmatrix}
        \bPhi \bLambda^{k-n} \bb \\
        \bPhi \bLambda^{k-n+1} \bb \\
        \cdots \\
        \bPhi \bLambda^k \bb
    \end{bmatrix}. 
\end{split}
\end{equation}
That is, $\bh_k$ is built by $(n+1)$-time delay snapshots of the observables. Accordingly, we use $(n+1)$-time delay observed data
\begin{equation}    \label{eq:data}
    \widehat{\by}_k = \by_k+ \widehat{\beps}_k = \begin{bmatrix}
        \bg(x_{k-n}) \\
        \bg(x_{k-n+1}) \\
        \vdots \\
        \bg(x_k)
    \end{bmatrix} + 
    \begin{bmatrix}
        \widehat{\epsilon}_{k-n} \\
        \widehat{\epsilon}_{k-n+1} \\
        \vdots \\
        \widehat{\epsilon}_k
    \end{bmatrix}.  
\end{equation}

Here $\widehat{\epsilon}_k$ is the noise in the observed data at time $k$. Note we use a different notation $\widehat{\epsilon}_k$ instead of $\bvareps_k$ to distinguish one realization of the noise and a random noise. More generally, $\widehat{\epsilon}_k$ can also include the model misspecification and measurement error etc.  

We will apply the EnKF algorithm \cref{eq:EnKF} to the above extended underlying system to obtain the Koopman modes and eigenvalues. Until the last step of filtering ($k=m$), all snapshots have been utilized to correct the estimates of Koopman modes and eigenvalues. Denote the posterior mean of $\btheta_m$ as $\overline{\btheta}^{+}_m$ and the covariance matrix as $P_m$. Then we can estimate the parameter $\btheta$ by
\[
    \btheta\thicksim\mathcal{N}(\overline{\btheta}^{+}_m,P_{m}).
\]
Particularly, the distribution of the DMD modes and eigenvalues are given by
\begin{equation}
\label{eigenpairs}
\begin{bmatrix}
\rm
\bm{\phi}_1\\
\vdots\\
\bm{\phi}_r
\end{bmatrix}
\thicksim\mathcal{N}(\begin{bmatrix}
\rm
\overline{\bm{\phi}}^{+}_{1}\\
\vdots\\
\overline{\bm{\phi}}^{+}_{r}
\end{bmatrix},P^{(1)}_{m}),\quad
\begin{bmatrix}
\rm
\lambda_1 \\
\vdots\\
\lambda_r
\end{bmatrix}
\thicksim\mathcal{N}(\begin{bmatrix}
\rm
\overline{\lambda}^{+}_{1} \\
\vdots\\
\overline{\lambda}^{+}_{r}
\end{bmatrix},P^{(2)}_{m}).
\end{equation}
In \cref{eigenpairs}, $\{\overline{\bm{\phi}}^{+}_{i},\overline{\lambda}^{+}_{i}\}$ are the components of $\overline{\btheta}^{+}_m$,  and $P^{(1)}_{m}$ and $P^{(2)}_{m}$ are the covariance matrices for the corresponding components. 
The complete algorithm of EnKF-DMD is shown in \cref{alg:2}.

\subsubsection{Initialization} 
It is critical to determine the initial value and truncated number for this method. We use the output by the exact DMD as initialization in our algorithm. That is, we apply \cref{alg:1} to the snapshots $\{y_0,y_1,\cdots,y_m\}$ to obtain an initial estimate of the DMD modes $\{\bm{\phi}^{\rm DMD}_i\}_{i=1}^r$, eigenvalues$\{\lambda^{\rm DMD}_i\}_{i=1}^r$, and coefficients $\{b^{\rm DMD}_i\}_{i=1}^r$. We take the truncated number $r$ to be the rank in SVD in the exact DMD. Explicitly,  
\[
    \btheta^{\rm DMD}=\begin{bmatrix}
    {\bm{\phi}^{\rm DMD}_1}\matT ,\cdots,{\bm{\phi}^{\rm DMD}_r}\matT ,\lambda_1^{\rm DMD},\cdots,\lambda_r^{\rm DMD},b_1^{\rm DMD},\cdots,b_r^{\rm DMD}\end{bmatrix}\matT 
\]
and $\btheta_0\thicksim\mathcal{N}(\btheta^{\rm DMD},C_0)$. Here $C_0$ is the covariance of the prior distribution of $\btheta$. 

\begin{algorithm}
	\renewcommand{\algorithmicrequire}{\textbf{Input:}}
	\renewcommand{\algorithmicensure}{\textbf{Output:}}
	\caption{EnKF-DMD}
	\label{alg:2}
	\begin{algorithmic}[1]
		\REQUIRE dataset $\{y_0,y_1,\cdots,y_m\}$, covariance matrix $Q_k$ of system noise ($Q_k=0$ for autonomous systems),  noise magnitude $\sigma$, time-delay number $n$ 
		\ENSURE  Gaussian distribution of DMD modes and eigenvalues $\{\phi_i, \lambda_i\}$ or  $\{\phi^{t_k,t_0}_i,\lambda^{t_k,t_0}_i\}$ 
		\STATE Compute exact DMD and obtain the initial value  $\{\bm{\phi}^{\rm DMD}_i,\lambda_i^{\rm DMD},b_i^{\rm DMD}\}_{i=1}^r$ and rank $r$
		\STATE Construct an evolved state-space model \cref{eq:auto_evol} and \cref{eq:obs}
		\STATE Iterate $\btheta_k$ by Ensemble Kalman filter
     
            \FOR{$k=1$ to $m$}   

            \STATE $\mean{\btheta}_{k+1} = \mean{\btheta}_k + K_k \Brac{ \widehat{\by}_k - \mean{\bh_k(\btheta_k)} }$
            
            \STATE $P_{k+1} = P_k - P_k^{\btheta,\by} \Brac{ P_k^{\by,\by}+ \sigma^2 I }^{-1} P_k^{\by,\btheta} + Q$ 
            \ENDFOR
    
		\STATE Return $\btheta_m$ for autonomous systems and return $\btheta_k ~(k=1,2,\cdots)$ for non-autonomous systems
	\end{algorithmic}
\end{algorithm}

\subsection{Non-autonomous dynamical systems}
Next, we extend the method to non-autonomous dynamical systems. Recall that corresponding Koopman operator family with its spectral decomposition \cref{decomposition-non}. For simplicity of discussion, we will set $t_0\equiv 0$ and introduce 
\[
    \lambda_i(k)=\Brac{\lambda_i^{t_k,0}}^{\frac{1}{k}},\quad  b_i(k)=\varphi^{t_k,0}_i(x_0),\quad   \phi_i(k)=\phi^{t_k,0}_i(k). 
\]
Then \cref{decomposition-non} can also be written as 
\begin{equation}
\label{eq:spectral}
    \bg(x_{k}) = \sum_{i=1}^{\infty} \Brac{\lambda_i(t_k)}^k \phi_{i}(k) b_i(k).
\end{equation}
We still write the relevant parameters together as
\begin{equation}    \label{eq:para_NonAuto}
    \btheta_k = \begin{bmatrix}
    \bphi\matT_1(k), \cdots, \bphi\matT_r(k), \lambda_1(k), \cdots, \lambda_r(k), b_1(k), \cdots, b_r(k)
    \end{bmatrix}\matT.
\end{equation}
Compared with the autonomous case, it can be seen that $\btheta_k$ changes with $k$ in the non-autonomous systems, while it remains constant in the autonomous scenario.

On the other hand, it should be noticed that in many applications $\Delta t=t_{k+1}-t_k$ is usually small, since the observables are frequent. Note that $\mathcal{K}^{t_k+1,0}=\mathcal{K}^{t_{k+1},t_k}\circ\mathcal{K}^{t_k,0}$, where $\mathcal{K}^{t_{k+1},t_k}$ describes the changes take place in $[t_k, t_{k+1}]$, which are often small. Therefore, $\mathcal{K}^{t_{k+1},t_k}\approx \mathcal{I}$, or  equivalently $\mathcal{K}^{t_{k+1},0}\approx \mathcal{K}^{t_k,0}$. Then by operator perturbation theory, see \cite{kloeckner2019effective,kato2013perturbation}, their spectral properties, in particular $\btheta_k$ and $\btheta_{k+1}$, will also be close to each other. 
In particular, we can write
\begin{equation}
\label{tra_non}
\btheta_{k+1} = \btheta_k+\delta \btheta_k.
\end{equation}

When $f(x,t)$ in \cref{non-sys} is available,  formulas for $\delta\btheta$ can be derived using operator perturbation theory. However, in practice we often only have access to the observation data, and the perturbation formulas are in general very involved. It is of this reason, we model $\delta \theta_k\sim \mathcal{N}(0, Q_k)$. Such modelling allows us to apply the EnKF method to the non-autonomous case in a similar manner.  

The observation function is constructed similarly as in the autonomous case: 
\begin{equation}
\label{obs_non}
\bh_k(\bm{\theta}_k)=\begin{bmatrix}
\bPhi(k)\bLambda^{k-n}(k)\bb(k) \\
\bPhi(k)\bLambda^{k-n+1}(k)\bb(k) \\
\vdots \\
\bPhi(k)\bLambda^{k}(k)\bb(k)
\end{bmatrix}.    
\end{equation}

The EnKF can be applied to \cref{tra_non} and \cref{obs_non} to obtain the spectral estimation of discrete Koopman operator $\mathcal{K}^{t_{k},0}$, since its formula also allows nonzero $Q_k$. 
Again the method is also included in \cref{alg:2}.

\subsection{Choice of delay coordinates}
\label{sec3_4}
The time-delay number $n$ is critical for \cref{alg:2}. It determines how much historical data   we used to update the parameter estimates. In general, one needs a sufficiently large $n$ so that $\bh_k$ is informative on the recovery of $\btheta$. In the Kalman filter literature \cite{MR3642295}, this is often characterized as the \emph{observability} of the system. 
On the other hand, if $n$ is large, the computational cost of $\bh_k$ will be high. To resolve this dilemma, we propose a method to determine $n$ adaptively according to the analysis of the residual of eigenvalues. We begin with the Koopman eigenvalues $\Lambda={\rm diag}(\lambda_1,\cdots,\lambda_r)$ of the local stencil of snapshots
\[\by_1,\by_2,\cdots,\by_{s},\]
where $s\geq r+1$. When new data $\by_{s+1}$ is available, we extend the stencil to \[\by_1,\by_2,\cdots,\by_{s},\by_{s+1},\]
and compute the Koopman eigenvalues $\Lambda'={\rm diag}(\lambda'_1,\cdots,\lambda'_r)$ of new stencil. Define the threshold of error as $\epsilon$. If the residual $e=\norm{\Lambda'-\Lambda}_{\rm F}$ is smaller than $\epsilon$, the number of delay coordinates is determined as $s+1$. Otherwise, input new data and repeat the above steps until $e\leq\epsilon$. The whole process to determine $n$ is concluded in \cref{tikz:1}. 

\thispagestyle{empty}
\tikzstyle{startstop} = [rectangle, rounded corners, minimum width = 2cm, minimum height=1cm,text centered, draw = black]
\tikzstyle{io} = [trapezium, trapezium left angle=70, trapezium right angle=110, minimum width=2cm, minimum height=1cm, text centered, draw=black]
\tikzstyle{process} = [rectangle, minimum width=3cm, minimum height=1cm, text centered, draw=black]
\tikzstyle{decision} = [diamond, aspect = 3, text centered, draw=black]
\tikzstyle{arrow} = [->,>=stealth]

\begin{figure}[htbp]
\centering
\begin{center}
\begin{tikzpicture}[node distance=0.5cm]
  \node[startstop]                        (start)   {\makecell{Data $\{\by_1,\cdots,\by_{r}\}$ and matrices $\bY_0$, $\bY_1$ \\ Koopman eigenvalues $\Lambda=\rm diag\{\lambda_1,\cdots,\lambda_{r}\}$\\ Begin with $s=r+1$ }};
  \node[process, below=of start]                         (step 1)  {Data $\{\by_1,\cdots,\by_{s}\}$};
  \node[process, below=of step 1]                       (step 2)  {Koopman eigenvalues $\Lambda'$};
  \node[decision, diamond, aspect=2, below=of step 2]     (choice)  {$\norm{\Lambda'-\Lambda}_{\rm F}\leq \epsilon$};
  \node[process, right=30pt of choice]                   (step x)  {\makecell{Add snapshot $\by_{s+1}$\\ Set $\Lambda=\Lambda'$ and $s=s+1$}};
  \node[startstop, below=20pt of choice]  (end)     {$n=s$};
  \draw[->] (start)  -- (step 1);
  \draw[->] (step 1) -- (step 2);
  \draw[->] (step 2) -- (choice);
  \draw[->] (choice) -- node[left]  {Yes} (end);
  \draw[->] (choice) -- node[above] {No}  (step x);
  \draw[->] (step x) -- (step x|-step 1) -> (step 1);
\end{tikzpicture}
\end{center}
\caption{Flowchart for determining the numeber of delay coordinates.} \label{tikz:1}
\end{figure}
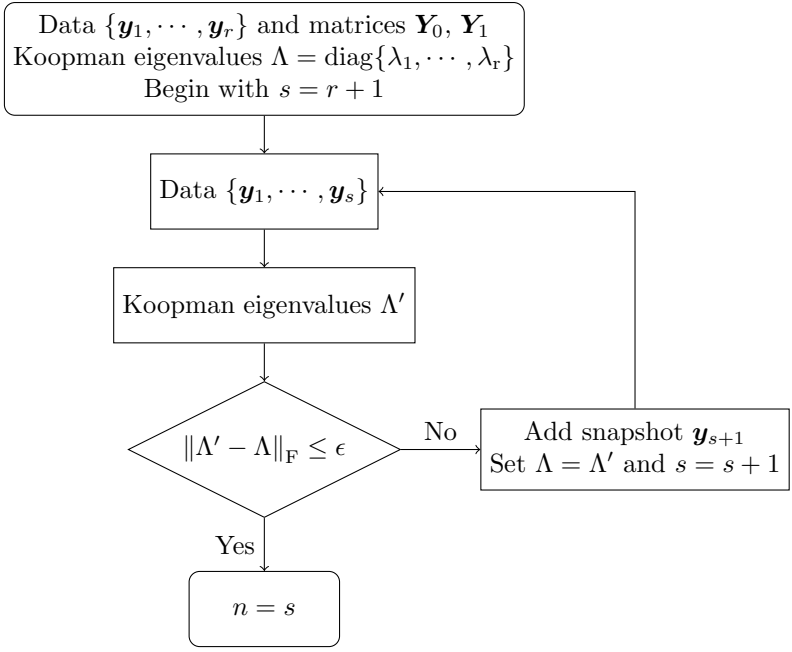

\section{Analysis result}
\label{sec4}
In this section, we will provide some theoretical analysis of the methods we have derived. 
\subsection{Autonomous system}
For autonomous system, we assume there is a ``true" finite dimensional Koopman parameter $\btheta^\dagger$, and the observation \cref{eq:data} are generated by it 
\begin{equation}    \label{eq:Auto_Obs}
    \widehat{\by}_k = \bh_k(\btheta^\dagger) + \widehat{\beps}_k. 
\end{equation}
This is a simplification as in general $\theta^\dagger$ is infinite dimensional, and our assumption can be seen as a finite dimensional truncation of it. Our analysis result relies on the following linearization assumption 
\begin{asm} \label{asm:Auto_Lin}
There exists $H_k \in \mR^{D\times d} $ and norm bound $M>0$ s.t. 
\begin{equation}    \label{eq:Auto_Lin}
    \bh_k(\btheta_k^{(i)}) = \bh_k(\btheta^\dagger) + H_k(\btheta_k^{(i)} - \btheta^\dagger), \quad \norm{H_k} \leq M. 
\end{equation}
\end{asm}

\begin{rem}
The linearization assumption holds in many nonlinear scenarios, especially when the parameter dimension is high and the ensemble is not large \cite{2308.16784}. To determine $H_k$, there are $Dd$ free unkowns and $DN$ linear constraints. We usually take $N<d$, so that there will always be a solution when these linear constraints are compatible. On the other hand, $H_k$ can be viewed as a local linear approximation of $\bh_k$. Since the ensemble members get closer as they evolve, such approximation will be good enough to ensure \cref{eq:Auto_Lin}, and $M$ can be chosen as the norm of $\bh_k$. 
\end{rem}

Our next theorem discusses the performance of ETKF in the long run. Theoretically speaking, ETKF is easier to analyze since its covariance update is deterministic. EnKF implementation with random $\xi^{(i)}_k$ in \cref{eq:EnKF} is often easier with similar empirical performance, but their analysis is much more involved due to random covariance updates and hence is omitted here. 
\begin{theorem} \label{thm:Auto_Converge}
Consider the ETKF method \cref{eq:ETKF_mean}, \cref{eq:ETKF_cov} for Koopman spectral estimation. Under \cref{asm:Auto_Lin}, and assume that the parameters are uniformly bounded: 
\begin{equation}    \label{eq:Asm_Auto_Bound}
    \max \Big\{ \normeo{\btheta^\dagger} , \sup_{k,i} \normeo{\btheta_k^{(i)}} \Big\} \leq R.
\end{equation}
Then there exist constants $C_1,C_2$ depending on $\sigma,M,R,Q$ and initialization, s.t. 
\begin{equation}    \label{eq:Auto_Convergence}
    \frac{1}{T} \sum_{k=0}^{T-1} \normeo{ \mean{\by}_k - \bh_k(\btheta^\dagger) }^2 \leq \frac{C_1}{T} + \frac{C_2}{T} \sum_{k=0}^{T-1} \normeo{\widehat{\beps}_k}.
\end{equation}
Here $\bar{y}_k$ is the averaged observation of the ensemble $\bar{y}_k = \frac{1}{N} \sum_{i=1}^N \bh_k(\btheta_k^{(i)})$.
\end{theorem}
The proof is delayed to \cref{app:Pf_Auto_Converge}. We point out that the above result holds deterministically, which mainly comes from that we do not introduce extra noise in the mean update. Note that $\normeo{ \mean{\by}_k - \bh_k(\btheta^\dagger) }^2$ is the ETKF data fitting error at time $k$. 
The result above essentially indicates that in long term average, ETKF estimates can fit observation data well. This is because the right hand side of \cref{eq:Auto_Convergence} has two terms. The first term $C_1/T\to 0$ when $T\to \infty$. The second term is the average size of the observation noise $|\hat{\epsilon}_k|$, which tends to be small for dynamical system applications. 

It is also worth comparing the results here with existing ones from inverse problem literature \cite{MR3988266,MR4089506,MR4405495,MR4487558}. 
Existing results often bound the parameter estimation error, i.e. $\|\theta_k-\theta^\dagger\|$. While such results are stronger, they often require very restrictive assumptions, for example the observation map is constant $\bh_k\equiv \bh$ and $\normeo{ \bh(\btheta) - \bh(\btheta^\dagger)}^2$  is a strongly convex function. Such requirements are very difficult to satisfy or verify for dynamical system related problems. 

\subsection{Non-autonomous system}
For non-autonomous system, recall that we assume the true Koopman parameter $\btheta_k^\dagger$ follows some unknown dynamic (see \cref{tra_non})
\begin{equation}    \label{eq:NonAuto_TrueEvol}
    \btheta_{k+1}^\dagger = \btheta_k^\dagger + \delta \btheta_k. 
\end{equation}
As a remark, we do not require $\delta\btheta_k$ to be generated from $\mathcal{N}(0, Q_k)$, which was the modeling assumption to implement EnKF. Similar to the autonomous case, one can write the observation as 
\begin{equation}    \label{eq:NonAuto_Obs}
    \widehat{\by}_k = \bh_k(\btheta_k^\dagger) + \widehat{\epsilon}_k. 
\end{equation}

For non-autonomous system, we also assume the adapted linearization assumption
\begin{asm}  \label{asm:NonAuto_Lin}
There exists $H_k \in \mR^{D\times d} $ and norm bound $M$ s.t. 
\begin{equation}    \label{eq:NonAuto_Lin}
    \bh_k(\btheta_k^{(i)}) = \bh_k(\btheta_k^\dagger) + H_k(\btheta_k^{(i)} - \btheta_k^\dagger), \quad \norm{H_k} \leq M. 
\end{equation}
\end{asm}

\begin{theorem} \label{thm:NonAuto_Converge}
Consider the ETKF method \cref{eq:ETKF_mean}, \cref{eq:ETKF_cov} for Koopman spectral estimation. Under \cref{asm:NonAuto_Lin}, and assume that the parameters are uniformly bounded
\begin{equation}    \label{eq:Asm_NonAuto_Bound}
    \max \Big\{ \sup_k \normeo{\btheta_k^\dagger} , \sup_{k,i} \normeo{\btheta_k^{(i)}} \Big\} \leq R.
\end{equation}
Then there exist constants $C_1,C_2,C_3$ depending on $\sigma,M,R,Q$ and initialization, s.t. 
\begin{equation}    \label{eq:NonAuto_Convergence}
    \frac{1}{T} \sum_{k=0}^{T-1} \normeo{ \mean{\by}_k - \bh_k(\btheta_k^\dagger) }^2 \leq \frac{C_1}{T} + \frac{C_2}{T} \sum_{k=0}^{T-1} \normeo{\widehat{\beps}_k} + \frac{C_3}{T} \sum_{k=0}^{T-1} \normeo{ \delta \btheta_k } .
\end{equation}
\end{theorem}
The proof is delayed to \cref{app:Pf_NonAuto_Converge}. Similar to the autonomous case, our result shows EnKF is effective in long term data fitting. The upper bound \cref{eq:NonAuto_Convergence} has an additional term compared with \cref{eq:Auto_Convergence}, which is the average temporal change of the Koopman parameters. It is often small if the dynamical system evolves slowly or $\Delta t$ is relatively small.

\section{Numerical results}
\label{sec5}
In this section, we implement EnKF-DMD in different experiments to verify its effectiveness in extracting the Koopman modes and eigenvalues. \Cref{sec5_1} shows the performance of EnKF-DMD in a low-dimensional autonomous ordinary differential equation. The method is also applied to a high-dimensional autonomous system in \Cref{sec5_2}. In \Cref{sec5_3}, we estimate time-varying eigenpairs in a linear non-autonomous system. An example of a 2-dimensional partial differential equation is experimented with in \Cref{sec5_4} to reconstruct the solution for autonomous and non-autonomous systems.

\subsection{Autonomous ordinary differential equation}
\label{sec5_1}
We consider the following nonlinear ordinary differential equation in two variables:

\begin{equation}
\label{low-audim}
\left\{
\begin{aligned}
\dot{x}_1&=\mu x_1, \\
\dot{x}_2&=\lambda(x_2-x_1^2),
\end{aligned}
\right.
\end{equation}
where $\mu=-0.01$ and $\lambda=-0.5$. An observable subspace spanned by the measurements $\{x_1,x_2,x_1^2\}$ is a Koopman invariant subspace of the system \cref{low-audim}, and $\{x_1,x_2,x_1^2\}$ are called dictionary functions.  Then we obtain a linear dynamical system on these observables that advance the original state
$x$:
\begin{equation}
\label{linear-audim}
\diff{}{t}\begin{bmatrix}
x_1 \\
x_2\\
x_1^2
\end{bmatrix}=
\mathbf{A}\begin{bmatrix}
x_1 \\
x_2\\
x_1^2
\end{bmatrix}
=
\begin{bmatrix}
\mu & 0 & 0 \\
0 & \lambda & -\lambda\\
0 & 0 & 2\mu
\end{bmatrix}
\begin{bmatrix}
x_1 \\
x_2\\
x_1^2
\end{bmatrix}.
\end{equation}
The discrete-time Koopman operator corresponding to the system \cref{linear-audim} is defined by $\exp (\mathbf{A}\Delta t)$, where $\Delta t=1$ is the discrete time step. We apply DMD on the extended observables and estimate Koopman eigenvalues and modes. The eigenvalues calculated by DMD with a noise-free dataset are shown in \cref{Tab:table1} together with true eigenvalues. Now we consider the extended observables with noise, such that 
\[
\by_k \thicksim \mathcal{N}\bigg( \begin{bmatrix}
x_1 \\
x_2\\
x_1^2
\end{bmatrix},  \sigma^2\mathbf{I}\bigg).
\]
When the dataset has a noise with magnitude $\sigma=0.1$, there is estimation error in eigenvalues from DMD methods, as shown in \cref{Tab:table1}. The EnKF-DMD obtains a more accurate result compared to standard DMD. 

\begin{table}[htbp]
\centering
\caption{\footnotesize{Value of eigenvalues.}}
\begin{tabular}{ccccccc}
\hline
\small Eigenvalues  &  DMD-noise free & DMD-noise & EnKF-DMD \\ \hline

    0.9900 & 0.9900 & 0.9829 & 0.9825 \\
    0.9801 & 0.9802& 0.7821 & 0.9584 \\
    0.6065 & 0.6065 & 0.5549 & 0.6164 \\ \hline
\end{tabular}
\label{Tab:table1}
\end{table}

The observable can be reconstructed with updated estimation of eigenvalues and modes. According to \cref{low-audim}, we reconstruct the solution in the time interval $[0,100]$. It can be seen that the reconstruction by EnKF-DMD approximates truth well in \cref{Fig:exp01_rec}. 

\begin{figure}[htbp]
  \centering  \includegraphics[width=4.5in, trim = 120 0 120 0 clip]{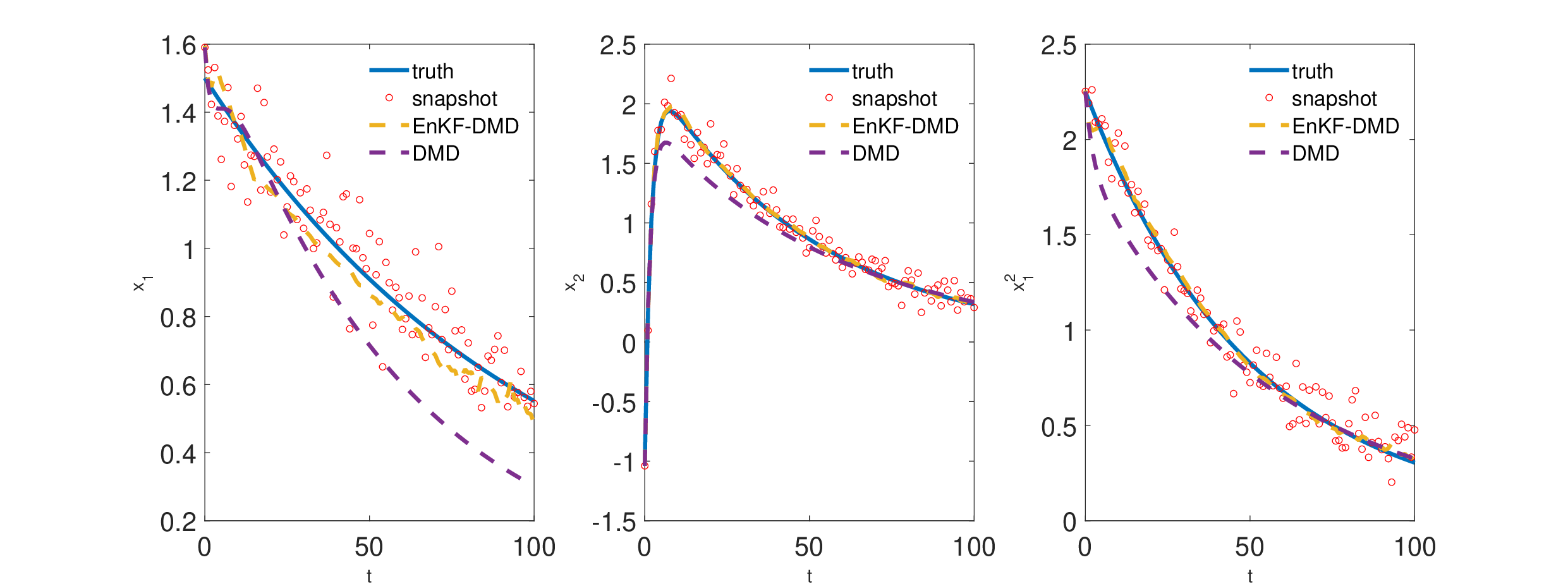}
  \caption{Reconstruction of observable functions.}
  \label{Fig:exp01_rec}
\end{figure}

We also investigate the application of EnKF-DMD in prediction. \cref{Fig:exp01_rec_pre} shows both the results of the reconstruction and prediction of EnKF-DMD. The construction is simulated in the time interval $[0,100]$ while the prediction is realized in $[100,200]$. The inaccurate estimation of $x_1$ at $t=100$ leads to a deviation of $x_1$ to truth in prediction. 

\begin{figure}[htbp]
  \centering
  \includegraphics[width=4.5in, trim = 120 0 120 0 clip]{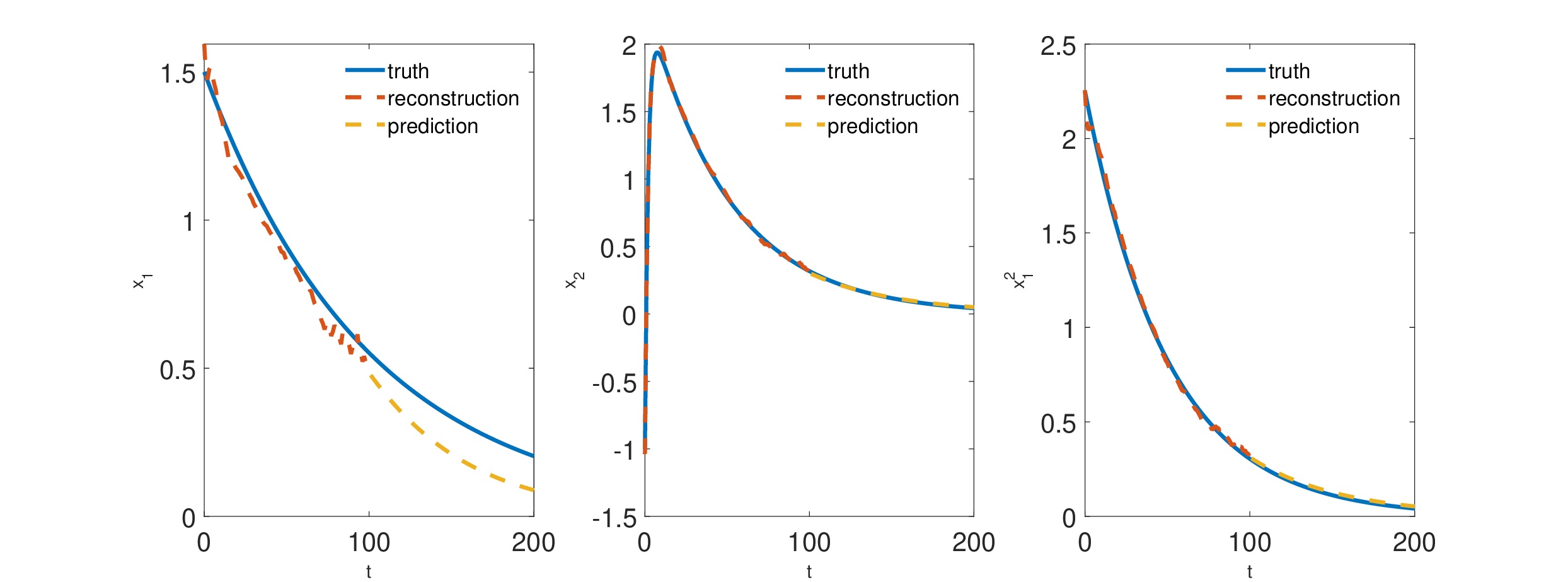}
  \caption{Reconstruction and prediction of observable functions.}
  \label{Fig:exp01_rec_pre}
\end{figure}
The root mean square error (RMSE) between the solution $\widehat{x}$ and the truth $x$ is computed by 
$${\rm RMSE}=\sqrt{\frac{\sum_{i=1}^{N}\norm{ \widehat{x}^i-x^i}^2}{N}}.$$ 
We compare the RMSE of EnKF-DMD and DMD-noise in \cref{Tab:table2}, which shows that EnKF-DMD has higher accuracy in the reconstruction and prediction of the solution.

\begin{table}[htbp]
\centering
\caption{\footnotesize{RMSE of EnKF-DMD compared to DMD with noise snapshots.}}
\begin{tabular}{ccccccc}
\hline
\small & DMD-noise & EnKF-DMD\\ \hline

   reconstruction & 0.1522  & 0.0355  \\
   prediction & 0.1210  & 0.0668  \\ \hline
\end{tabular}
\label{Tab:table2}
\end{table}

\subsection{Sparse linear dynamical in the Fourier domain}
\label{sec5_2}
In this example, we design a high-dimensional system to examine the performance of EnKF-DMD. The system is created with $K=5$ nonzero two-dimensional spatial Fourier modes with other modes that are exactly zero. Then we define a stable linear time-invariant dynamical system on the $K$ modes. Randomly choose a temporal oscillation frequency and a small damping rate for each mode independently. A linear combination of coherent spatial Fourier modes, oscillating at a different fixed frequency, with eigenvalue $\lambda=d+\mi\omega$ constructs the system in the spatial domain. The parameters are determined in \cref{Tab:table3}.  Let
\begin{equation}
\label{sys:Fourier}
    \widehat{x}({\rm I},{\rm J},t)= \mee^{\lambda({\rm I,J}) t}= \mee^{dt}(\cos(\omega t)+\mi\sin (\omega t)),
\end{equation}
where ${\rm I}$ and ${\rm J}$ are positions randomly chosen for five modes, and $\lambda({\rm I,J})$ is corresponding eigenvalue. It is also possible to allow other Fourier modes to be contaminated with Gaussian noise and impose a fast and stable dynamic in each direction. Then
\[\widetilde{x}(t)=\widehat{x}(t)+\mathcal{N}(0,\sigma^2),\]
where $\sigma=10^{-3}$. We create snapshots by taking the real part of the inverse Fourier transform $\mathcal{F}^{-1}$:
\[
x(t)={\rm real}(\mathcal{F}^{-1}(\widetilde{x}(t))).
\]
\begin{table}[htbp]
	\centering
	\setlength{\abovecaptionskip}{0cm}
	\setlength{\belowcaptionskip}{0.2cm}
	\caption{Parameters of the system \cref{sys:Fourier}.}
	\scalebox{0.8}{
		\begin{tabular}{cccccc}
			\toprule
			Modes & 1 & 2 & 3 & 4 & 5 \\
			\midrule
			& & & & & \\[-6pt]
			(I, J) & (3, 2) & (5, 2)  & (7, 9)  & (3, 5)  & (9, 8)  \\
			\cline{1-6}
			& & & & & \\[-6pt]
			 d={\rm real}($\lambda$) &-0.0755 & -0.0839 & -0.0414 & -0.0175 & -0.0702\\
                \cline{1-6}
			& & & & & \\[-6pt]
              $\omega$={\rm imag}($\lambda$) & 12.0967 & 11.0315 & 8.1959 & 6.2861 & 1.6272\\
                \cline{1-6}
			& & & & & \\[-6pt]
              Initial condition &0.2954 &-0.7263 &-0.4689 &0.3091 & 0.3857 \\
			\bottomrule
		\end{tabular}
	}
        \label{Tab:table3}
\end{table}
\cref{Fig:fourier_sys} shows that Fourier mode coefficients generate distinct spatial coherent patterns and contribute to the spatial structures. It displays the spatial structures of the system with noise level $\sigma=0.004$. \cref{Fig:FFT} shows real and imaginary parts of five modes. We use compressed DMD \cite{brunton2016compressed} in this example for comparison to accelerate the process of DMD in high dimensions. Compressed DMD is a method that computes DMD on the compressed data of original full-state snapshot data. Then the full-state DMD modes are reconstructed according to the compressed DMD transformations. Compressed DMD does not extract the spatial modes in the case of noisy snapshots, as shown in \cref{Fig:cDMD}. The mode reconstruction is displayed in \cref{Fig:oDMD}. As seen in \cref{Fig:oDMD}, the EnKF-DMD works extremely well even if there is background noise. The EnKF-DMD starts with compressing the full-state snapshots, performing DMD, and then updating the modes in the filtering. Suppose that the compressed datasets are $Y_0'$ and $Y_1'$. Then we apply \cref{alg:1} to the data and obtain $\bm{V}'_r$,$\bm{\Sigma}_r'$ and $\widetilde{\bm{\phi}}'$. The reconstruction of modes of $Y_0$ and $Y_1$ is given by
\[\bm{\phi}=Y_1\bm{V}'_r\bm{\Sigma}_r'^{-1}\widetilde{\bm{\phi}}'.\]
The eigenvalues estimated by EnKF-DMD match the true eigenvalues better than that by compressed DMD in \cref{Fig:fourier_eig}. But there still exists a deviation between our estimates and the truth. In this example, it is important to determine the number of time-delay coordinates to avoid high computational cost. We use the adaptive method in \Cref{sec3_4} to determine delay coordinates.

\begin{figure}[htbp]
  \centering
  \includegraphics[width=4.5in]{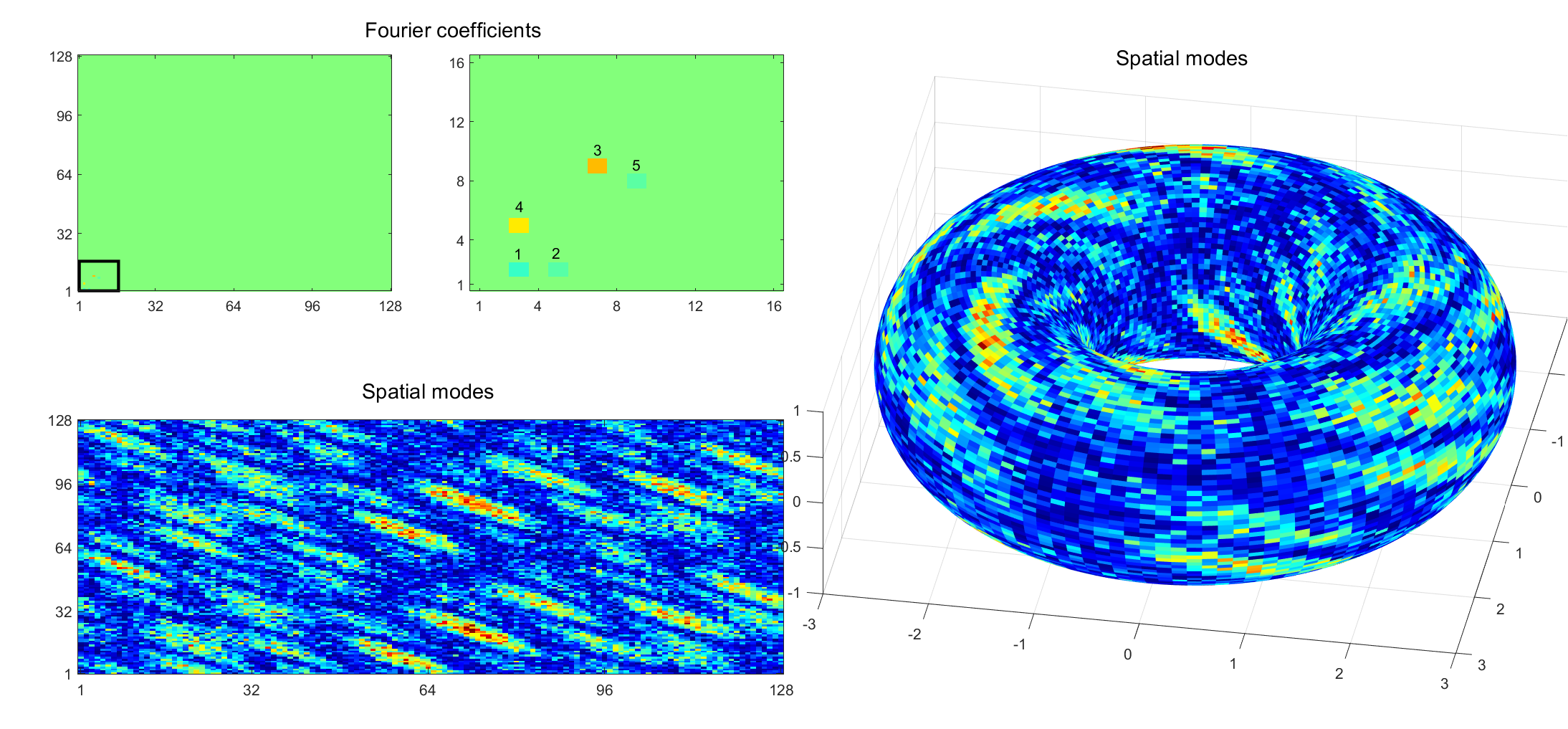}
  \caption{Five coefficients in the Fourier domain and spatial structures of the dynamical system.}
  \label{Fig:fourier_sys}
\end{figure}

\begin{figure}[htbp]
  \centering
  \subfigure[True modes]{
  \includegraphics[width=4in]{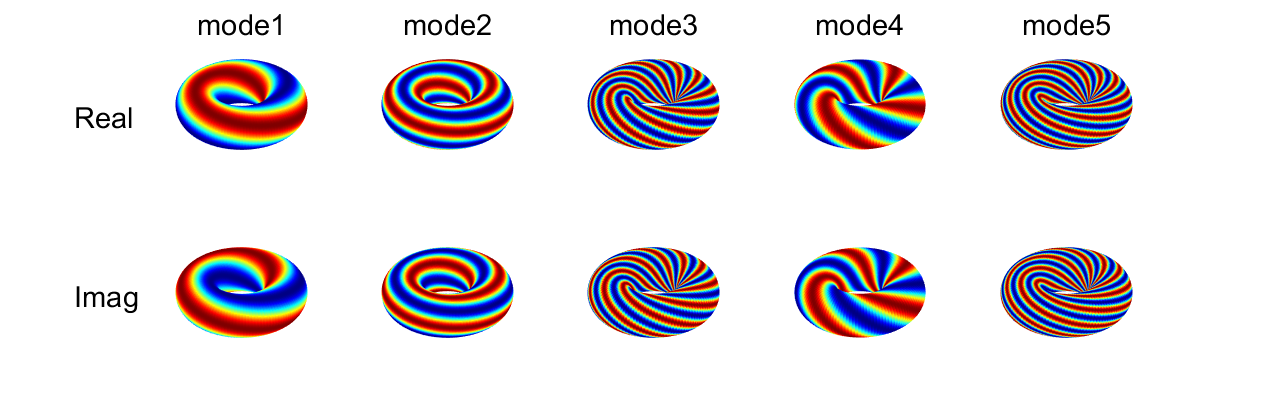}
  }
  \label{Fig:FFT}
  \quad
  \subfigure[Modes by compressed DMD]{
  \includegraphics[width=4in]{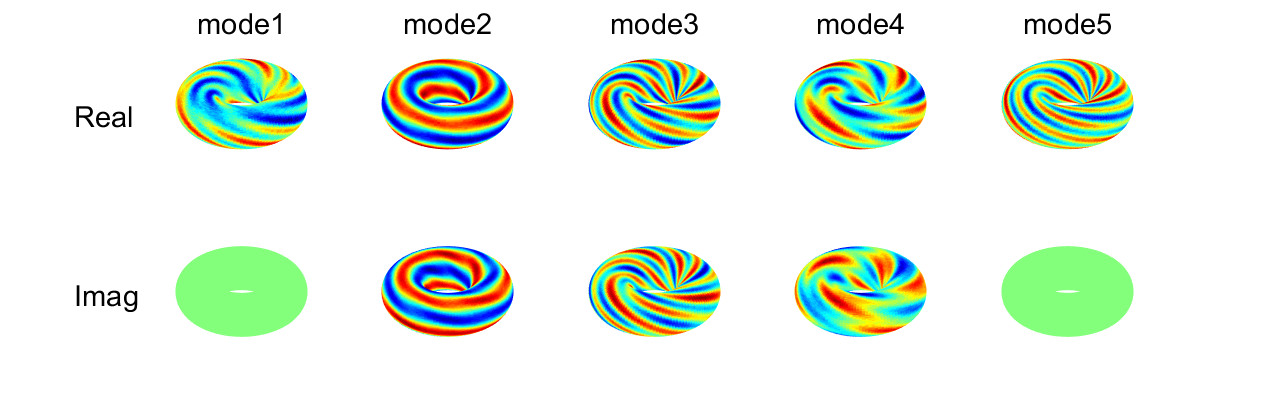}
  }
  \label{Fig:cDMD}
  \quad
  \subfigure[Modes by EnKF-DMD]{
  \includegraphics[width=4in]{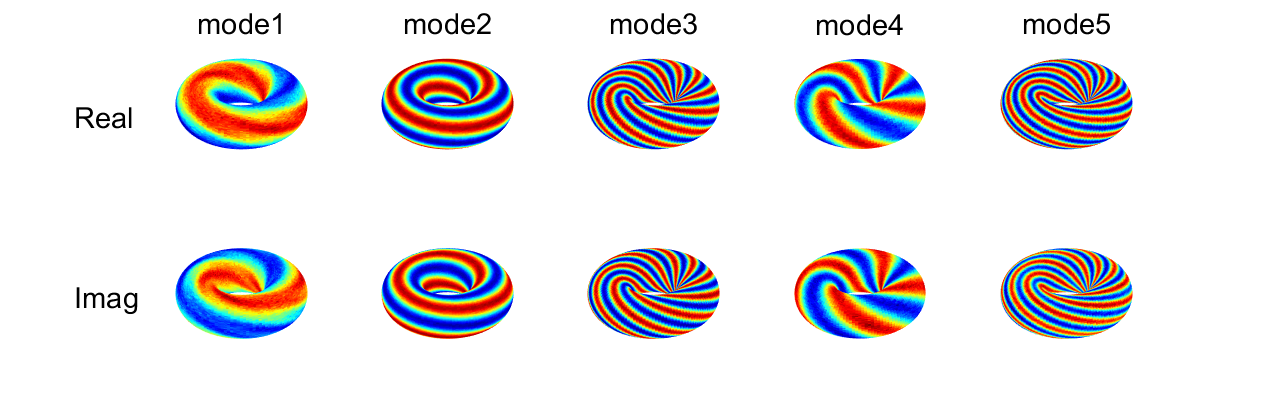}
  }
  \label{Fig:oDMD}
  \caption{True DMD modes, modes estimated by compressed DMD and EnKF-DMD.}
\end{figure}

\begin{figure}[htbp]
  \centering
  \includegraphics[width=3in]{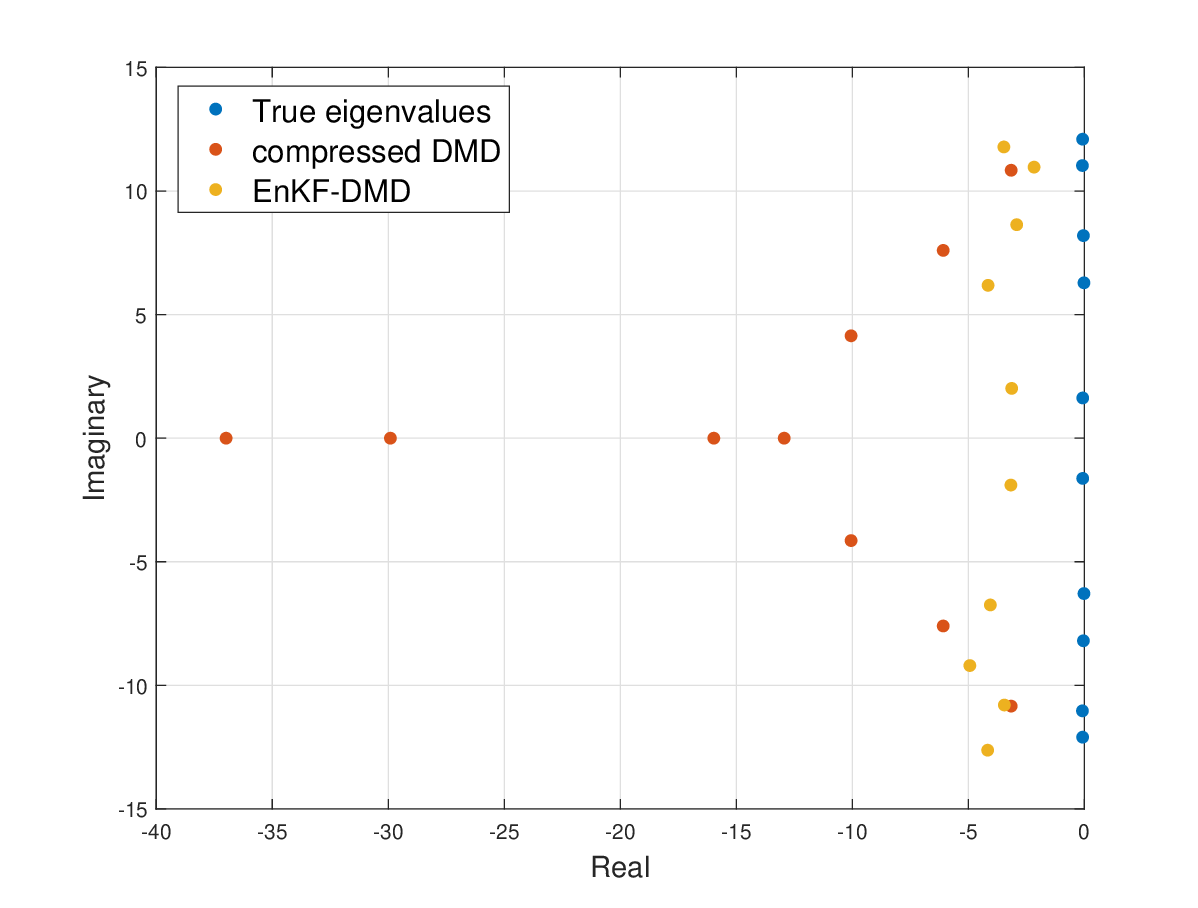}
  \caption{True eigenvalues, eigenvalues by compressed DMD and EnKF-DMD.}
  \label{Fig:fourier_eig}
\end{figure}

\subsection{Linear non-autonomous system}
\label{sec5_3}
Consider a non-autonomous ordinary differential equation
\begin{equation}
\label{ode_nonauto}
\dot{\textbf{x}}=\textbf{A}(t)\textbf{x},
\end{equation}
where $\textbf{x}$ is the d-dimensional state and $\textbf{A}(t)$ is a time-dependent matrix. Let us determine
\begin{equation}
\textbf{A}(t) =
\begin{aligned}
\begin{pmatrix}
\sigma(t) & \omega(t)\\
-\omega(t) & \sigma(t)
\end{pmatrix}.
\end{aligned}
\end{equation}
The eigenvalues of the underlying matrices are $\lambda_1(t)=\sigma(t)+\omega(t) \mi$ and $\lambda_2(t)=\sigma(t)-\omega(t)\mi$. Set $\sigma(t)=\cos t$ and $\omega(t)=2$. In this case, the matrices are commutative and the fundamental matrix of the system can be analytically obtained by
\[
\mathcal{M}(t,t_0)= \mee^{\alpha(t,t_0)} \begin{pmatrix}
\cos \beta(t,t_0)  & \sin \beta(t,t_0)\\
-\sin \beta(t,t_0) & \cos \beta(t,t_0)
\end{pmatrix},
\]
where
\[\alpha(t,t_0)=\int_{t_0}^t  \sigma({\tau}) \mdd\tau \quad \text{and} \quad \beta(t,t_0)=\int_{t_0}^t \omega({\tau}) \mdd\tau.\]
Then the eigenvalues 
of matrix $\mathcal{M}(t,t_0)$ are $\mu(t,t_0)=\mee^{\alpha(t,t_0)\pm 2t\mi}$. We discrete the system with a uniform time interval $\Delta t=0.01$. 

If observable function satisfies $\bg(x)=x$, the fundamental matrix $\mathcal{M}(t_k,t_{0})$ is exactly Koopman operator $\mathcal{K}(t_k,t_{0})$. To see how the proposed algorithm performs in realistic situations, we add noise to the exact solution of the system \cref{ode_nonauto} and consider observables 
\[\by_k = \textbf{x}_k+\varepsilon,\quad \varepsilon\thicksim\mathcal{N}(0,\delta^2\mathbf{I}),\]
where $\delta=0.1$. The results in \cref{Fig:eig_c02} present the data with Gaussian noise of the standard deviation $\delta$. In the application of EnKF-DMD on the data,  we approximate time-varying eigenvalues of Koopman matrix $\mathcal{K}(t_k,t_{0})$. The results are displayed in \cref{Fig:eig_c01}, where we observe that the principal Koopman eigenvalues clearly show oscillation, and eigenvalues by EnKF-DMD fit the truth well. The solution of the system \cref{ode_nonauto} is also reconstructed by estimated eigenvalues and modes, as shown in \cref{Fig:eig_c02}. The method corrects the trajectory of the system from noisy snapshots. 

 \begin{figure}[htbp]
   \centering
   \includegraphics[width=4.5in, trim = 100 0 100 0 clip]{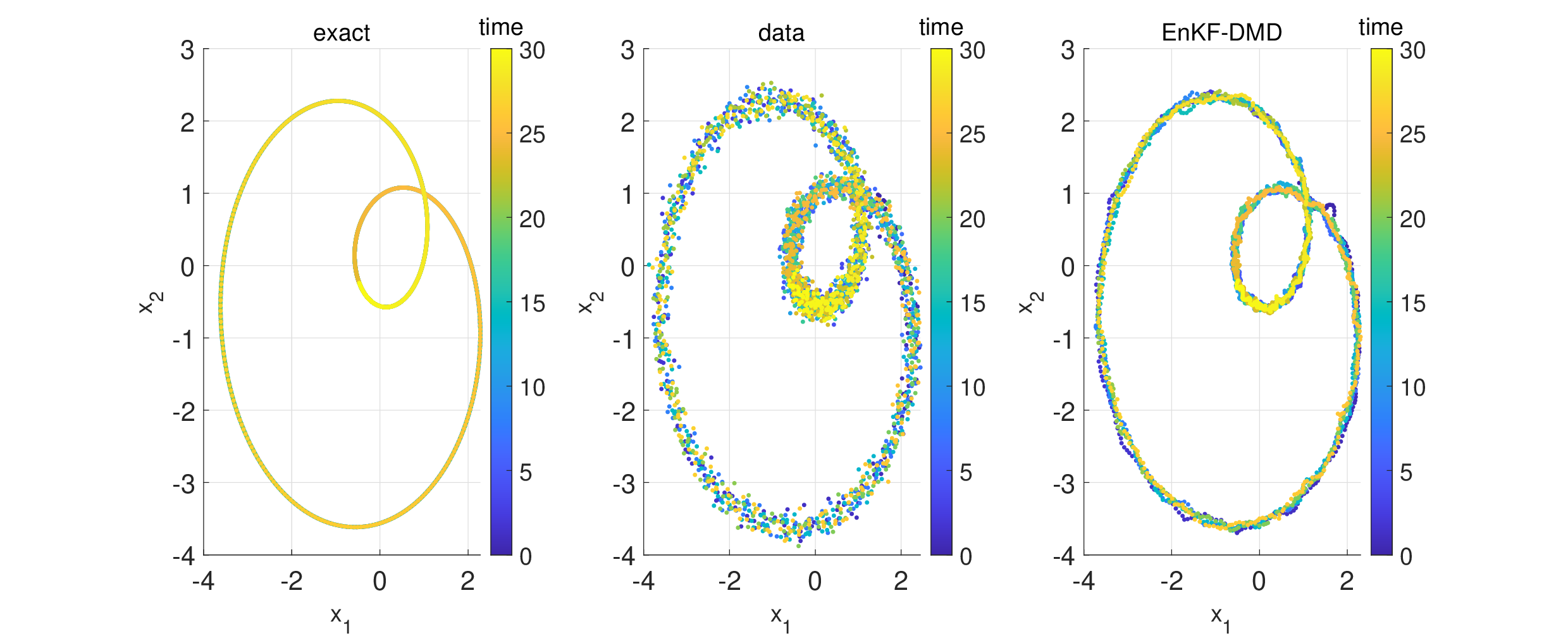}
   \caption{Trajectory of solution of system \cref{ode_nonauto} with $\sigma(t)=\cos(t)$.}
   \label{Fig:eig_c02}
 \end{figure}
 
 \begin{figure}[htbp]
   \centering
   \includegraphics[width=4.5in, trim = 100 0 100 0 clip]{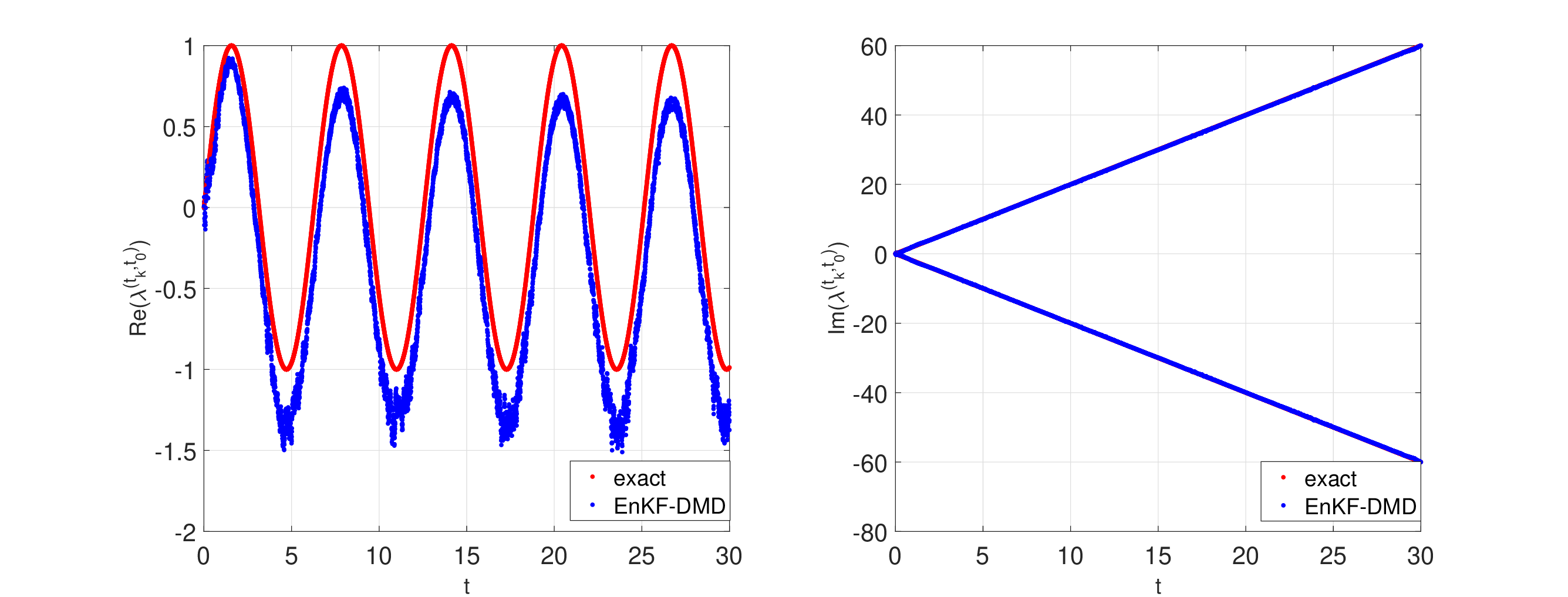}
   \caption{Real and Imaginary parts of Koopman eigenvalues with $\sigma(t)= \cos(t)$.}
   \label{Fig:eig_c01}
 \end{figure}

When we make $\sigma(t) = \cos(t/3)$ to decrease the frequency of the system \cref{ode_nonauto}, the performance of EnKF-DMD is not good anymore. In \cref{Fig:eig_c03}, there is a stable error between the eigenvalues estimated by EnKF-DMD and the truth in real parts. The reason comes from the algorithm itself because the time-delay observables only provide useful information in a specific short period. However, the solution reconstructed by the estimates improves accuracy compared to noisy data, as shown in \cref{Fig:eig_c04}.

 \begin{figure}[htbp]
   \centering
   \includegraphics[width=4.5in, trim = 100 0 100 0 clip]{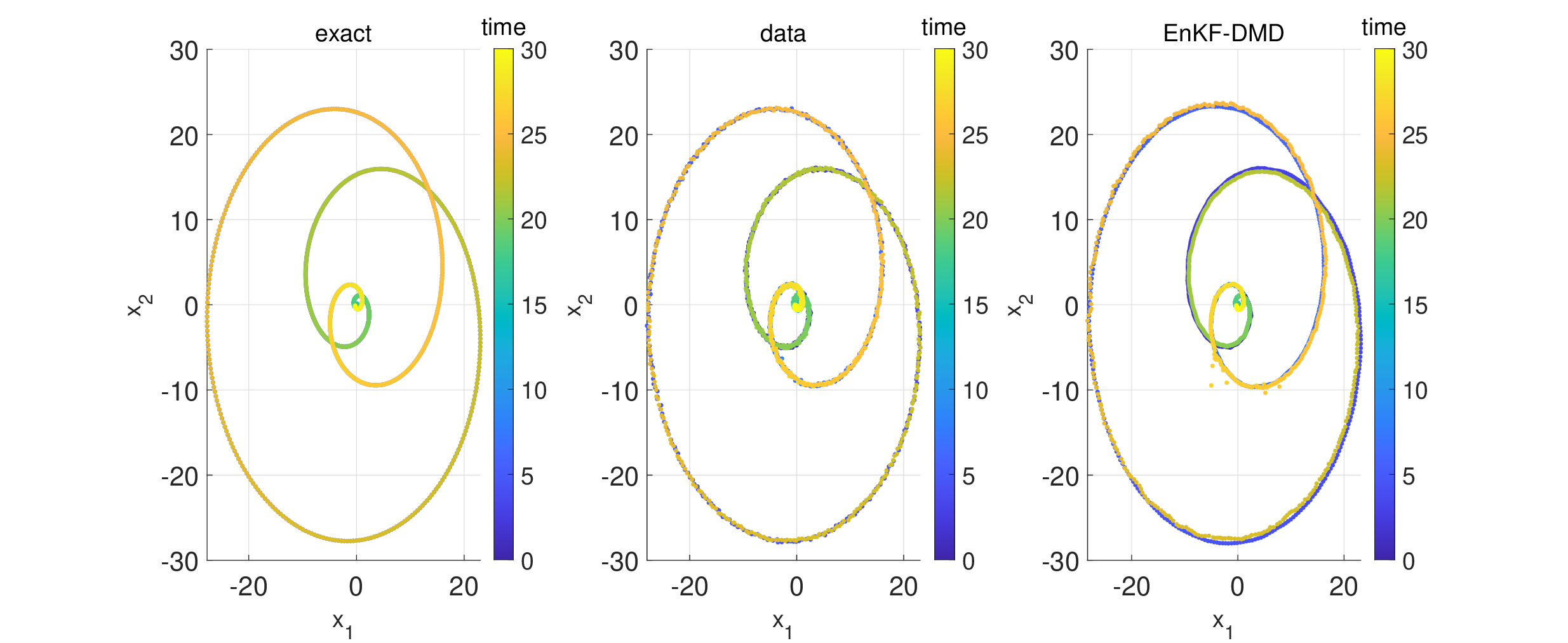}
   \caption{Trajectory of solution of system \cref{ode_nonauto} with $\sigma(t)=\cos(t/3)$.}
   \label{Fig:eig_c04}
 \end{figure}
 
  \begin{figure}[htbp]
   \centering
   \includegraphics[width=4.5in, trim = 100 0 100 0 clip]{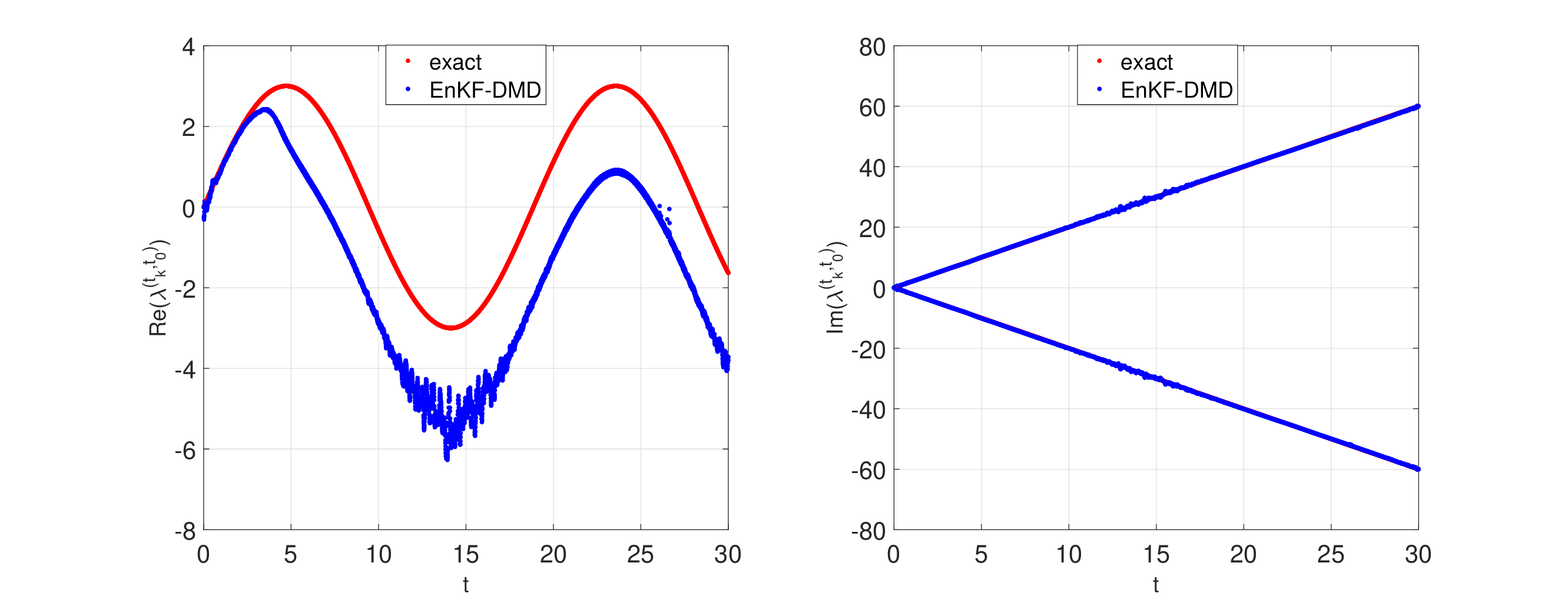}
   \caption{Real and Imaginary parts of Koopman eigenvalues with $\sigma(t)=\cos(t/3)$.}
   \label{Fig:eig_c03}
 \end{figure}

\subsection{Allen-Cahn equation}
\label{sec5_4}
The numerical example of this subsection is designed to test EnKF-DMD in a semilinear parabolic partial differential equation. We consider the Allen-Cahn equation
\begin{equation}
\label{ac_equation}
\begin{cases}
u_t = \nabla \cdot (\theta \nabla u) - \mu (u - u^3), & (t,x,y) \in [0,5] \times (0,2\pi) \times (0,2\pi), \\
u(x,y)|_{\Gamma_1} = 0, & \\
\frac{\partial u}{\partial n}(x,y)|_{\Gamma_2} = 0, &
\end{cases}
\end{equation}
where $\theta=0.1, \Gamma_1= \{0\} \times[0,2\pi] \bigcup [0,2\pi] \times \{0\} $ and $\Gamma_2= \{2\pi\} \times[0,2\pi] \bigcup [0,2\pi]\times \{2\pi\} $. The initial
condition is
$$u_0(x,y)=0.05\sin x\cdot \sin y.$$
In this example, we divide the dynamical system \cref{ac_equation} into autonomous and non-autonomous cases according to the function of $\mu$. First, when $\mu=1$, it determines an autonomous system.  We divide the spatial domain in each direction with a grid size $\Delta x=\Delta y=2\pi/50$.\ and discretize the temporal domain with a step size of $\Delta t=0.05$. The noisy snapshots are obtained by the exact solution with Gaussian noise $\varepsilon$ of noise magnitude $\sigma=0.1$. We generate snapshots in the interval $t\in[0,2]$. We use Extended DMD (EDMD) \cite{williams2015data}, a DMD method with extended observables, to extrapolate numerical solutions of this parabolic equation. The extended observable is $[u,u^3]$ in this example. We visualize the predicted solution of  EDMD and EnKF-DMD in \cref{Fig:ac_solution}. We only show the solution from $t=2$ to $t=5$, because the solution keeps stable after $t=5$. EDMD is not accurate at all since it is ineffective with noisy data. In comparison, EnKF-DMD achieves remarkably improved accuracy. We also determine the number of time-delay coordinates adaptively and take $n=6$.

\begin{figure}[htbp]
  \centering
  \includegraphics[width=4in, trim = 120 40 120 0 clip]{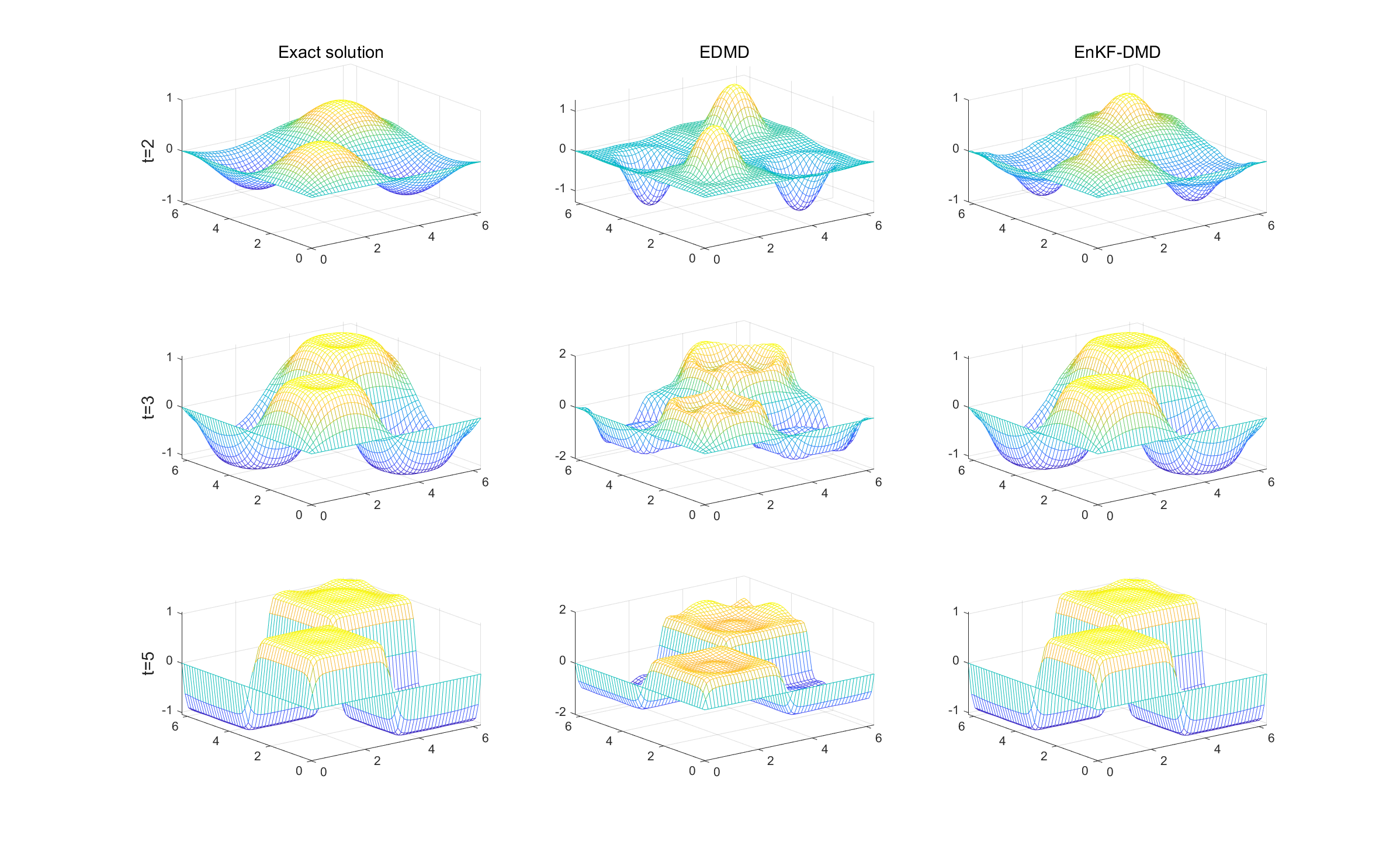}
  \caption{Predicted solution of system \cref{ac_equation} at different time with $\mu=1$.}
  \label{Fig:ac_solution}
\end{figure}

Let $\mu(t)=\sin (t)$, then the system \cref{ac_equation} becomes a non-autonomous system. The solutions for EnKF-DMD are presented in \cref{Fig:ac_solution_non}. The method extracts time-varying structures and models the system well in the case of noisy data.  

\begin{figure}[htbp]
  \centering
  \includegraphics[width=4in, trim = 120 40 120 0 clip]{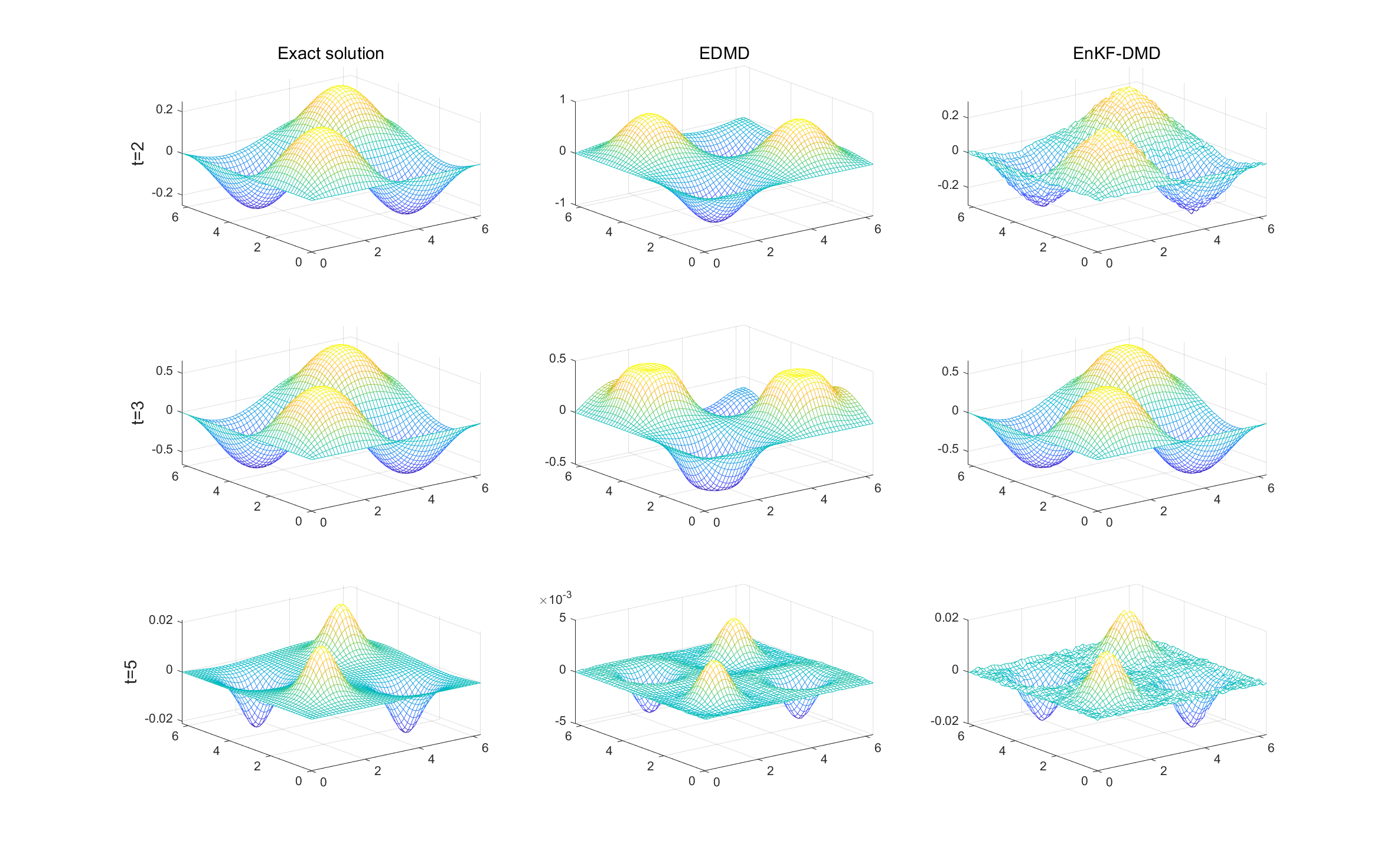}
  \caption{Predicted solution of system \cref{ac_equation} at different time with $\mu(t)=\sin (t)$.}
  \label{Fig:ac_solution_non}
\end{figure}

\section{Conclusion}    \label{sec6}
In this paper, we proposed an EnKF-DMD method that uses the Kalman filter to estimate Koopman modes and eigenvalues. Due to the advantage of EnKF, the method was applied to improve the accuracy of DMD algorithms with noisy snapshots. We use an identity map as a forecast map for modes and eigenvalues, which are the states to be estimated, as they are constants for autonomous systems. For non-autonomous systems, the modes and eigenvalues are time-varying variables. We use a random walk model to approximate the underlying dynamics, so EnKF can be applied. 
Theoretical analysis was also provided for both cases of systems. We demonstrate the effectiveness of our methods numerically on some linear models where explicit solutions are available and the Allen-Cahn equation. 

\bibliographystyle{siam}
\bibliography{references}

\appendix
\section{Proofs}
\subsection{Proof for autonomous system}  \label{app:Pf_Auto_Converge}
We will use the Mahalanobis norm w.r.t. a positive definite matrix $P$, 
\begin{equation}    \label{eq:Mnorm}
    \norme{\btheta}_P^2 := \btheta\matT P^{-1} \btheta. 
\end{equation}

\begin{proof}[Proof for \cref{thm:Auto_Converge} ]
Under \cref{asm:Auto_Lin}, one can simplify 
\[
    \widehat{\by}_k - \bh_k(\btheta_k^{(i)}) = \bh_k(\btheta^\dagger) - \bh_k(\btheta_k^{(i)}) + \widehat{\beps}_k = - H_k ( \btheta_k^{(i)} - \btheta^\dagger ) + \widehat{\beps}_k.
\]
where we also use \cref{eq:Auto_Obs}. The mean update \cref{eq:ETKF_mean} can then be written as
\begin{equation}    \label{eq:Pf_Evolmean}
    \mean{\btheta}_{k+1} = \mean{\btheta}_k - K_k H_k ( \mean{\btheta}_k - \btheta^\dagger ) + K_k \widehat{\beps}_k,  
\end{equation}
or more compactly,
\[
    \mean{\btheta}_{k+1} - \btheta^\dagger = ( I - K_k H_k )( \mean{\btheta}_k - \btheta^\dagger ) + K_k \widehat{\beps}_k. 
\]
Taking the Mahalanobis norm $\normeo{\cdot}_{P_{k+1}}$ on both side, 
\begin{align*}
    & \normeo{ \mean{\btheta}_{k+1} - \btheta^\dagger }_{P_{k+1}}^2\\
    =~& \normeo{ ( I - K_k H_k )( \mean{\btheta}_k - \btheta^\dagger ) }_{P_{k+1}}^2 + \normeo{ K_k \widehat{\beps}_k }_{P_{k+1}}^2 + 2 \Brac{ K_k \widehat{\beps}_k }\matT P_{k+1}^{-1} ( I - K_k H_k )( \mean{\btheta}_k - \btheta^\dagger ).
\end{align*}
Therefore, 
\begin{equation}    \label{eq:pf_Decomp}
\begin{split}
    & \normeo{ \mean{\btheta}_{k+1} - \btheta^\dagger }_{P_{k+1}}^2 - \normeo{ ( I - K_k H_k )( \mean{\btheta}_k - \btheta^\dagger ) }_{P_{k+1}}^2 \\
    =~& \normeo{ K_k \widehat{\beps}_k }_{P_{k+1}}^2 + 2 \Brac{ K_k \widehat{\beps}_k }\matT P_{k+1}^{-1} ( I - K_k H_k )( \mean{\btheta}_k - \btheta^\dagger ) \\
    =~& \normeo{ K_k \widehat{\beps}_k }_{P_{k+1}}^2 + 2 \Brac{ K_k \widehat{\beps}_k }\matT P_{k+1}^{-1} ( \mean{\btheta}_{k+1} - \btheta^\dagger - K_k \widehat{\beps}_k ) \\
    =~& 2 \Brac{ K_k \widehat{\beps}_k }\matT P_{k+1}^{-1} ( \mean{\btheta}_{k+1} - \btheta^\dagger ) - \normeo{ K_k \widehat{\beps}_k }_{P_{k+1}}^2 \\
    \leq~& 2 \Brac{ K_k \widehat{\beps}_k }\matT P_{k+1}^{-1} ( \mean{\btheta}_{k+1} - \btheta^\dagger ). 
\end{split}
\end{equation}
The second term (corresponding the observation error) can be bounded by 
\begin{equation}    \label{eq:Pf_nos_ctrl}
\begin{split}
    2 \Brac{ K_k \widehat{\beps}_k }\matT P_{k+1}^{-1} &( \mean{\btheta}_{k+1} - \btheta^\dagger ) = 2 \widehat{\beps}_k \matT K_k\matT P_{k+1}^{-1} ( \mean{\btheta}_{k+1} - \btheta^\dagger ) \\
    \leq~& 2 \normeo{ \widehat{\beps}_k } \normo{ K_k\matT P_{k+1}^{-1} } \normeo{ \mean{\btheta}_{k+1} - \btheta^\dagger } \leq 4 R \sigma^{-2} M \normeo{ \widehat{\beps}_k }.
\end{split}
\end{equation}
where we use $ \normo{ K_k\matT P_{k+1}^{-1} } \leq \sigma^{-2} M$ by \cref{lem:KP_Ctrl} and $ \normeo{ \mean{\btheta}_{k+1} - \btheta^\dagger } \leq 2R $ by assumption \cref{eq:Asm_Auto_Bound}. Denote $W_k = P_k^{-1} - (I-K_k H_k)\matT P_{k+1}^{-1} (I-K_k H_k)$, then 
\begin{align*}
    &\normeo{ \mean{\btheta}_k - \btheta^\dagger }_{P_k}^2 - \normeo{ ( I - K_k H_k )( \mean{\btheta}_k - \btheta^\dagger ) }_{P_{k+1}}^2 \\ 
    =~& (\mean{\btheta}_k -  \btheta^\dagger)\matT \Rectbrac{ P_k^{-1} - (I-K_k H_k)\matT P_{k+1}^{-1} (I-K_k H_k) } ( \mean{\btheta}_k - \btheta^\dagger ) \\
    =~& (\mean{\btheta}_k -  \btheta^\dagger)\matT W_k ( \mean{\btheta}_k - \btheta^\dagger ).
\end{align*}
Plug it and \cref{eq:Pf_nos_ctrl} into \cref{eq:pf_Decomp}, one obtain
\begin{equation}    \label{eq:Pf_Decay}
    \normeo{ \mean{\btheta}_{k+1} - \btheta^\dagger }_{P_{k+1}}^2 \leq \normeo{ \mean{\btheta}_k - \btheta^\dagger }_{P_k}^2 - (\mean{\btheta}_k - \btheta^\dagger)\matT W_k ( \mean{\btheta}_k - \btheta^\dagger ) + 4 \sigma^{-2} M R \normeo{ \widehat{\beps}_k }.
\end{equation}
By \cref{lem:W_bound}, $W_k \succeq \mu^{-1} H_k\matT H_k$ for $\mu = 2 M^2 R^2 + \sigma^2$, so that
\[
    (\mean{\btheta}_k - \btheta^\dagger)\matT W_k ( \mean{\btheta}_k - \btheta^\dagger ) \geq \mu^{-1} \normeo{ H_k ( \mean{\btheta}_k - \btheta^\dagger) }^2 = \mu^{-1} \normeo{ \mean{\by}_k - \bh_k(\btheta^\dag) }^2, 
\]
where we use the relation $\mean{\by}_k - \bh_k(\btheta^\dag) = H_k ( \mean{\btheta}_k - \btheta^\dagger) $. So that \cref{eq:Pf_Decay} implies
\[
    \mu^{-1} \normeo{ \mean{\by}_k - \bh_k(\btheta^\dag) }^2 \leq \normeo{ \mean{\btheta}_k - \btheta^\dagger }_{P_k}^2 - \normeo{ \mean{\btheta}_{k+1} - \btheta^\dagger }_{P_{k+1}}^2 + 4 \sigma^{-2} M R \normeo{ \widehat{\beps}_k }.
\]
Taking summation over $k=0,1,\dots,T-1$ gives
\[
    \mu^{-1} \sum_{k=0}^{T-1} \normeo{ \mean{\by}_k - \bh_k(\btheta^\dag) }^2 \leq \normeo{ \mean{\btheta}_0 - \btheta^\dagger }_{P_0}^2 - \normeo{ \mean{\btheta}_T - \btheta^\dagger }_{P_T}^2 +  4 \sigma^{-2} M R \sum_{k=0}^{T-1} \normeo{ \widehat{\beps}_k }. 
\]
Denote $C_1 = (2 M^2 R^2 + \sigma^2) \normeo{ \mean{\btheta}_0 - \btheta^\dagger }_{P_0}^2 $ and $ C_2 = 4 \sigma^{-2} M R ( 2 M^2 R^2 + \sigma^2 )$, then
\[
    \frac{1}{T} \sum_{k=0}^{T-1} \normeo{ \mean{\by}_k - \bh_k(\btheta^\dag) }^2 \leq \frac{C_1}{T} + \frac{C_2}{T} \sum_{k=0}^{T-1} \normeo{ \widehat{\beps}_k }.
\]
This completes the proof.
\end{proof}
    
\begin{lemma} \label{lem:CovUpdate}
Under linearization assumption \cref{eq:Auto_Lin}, it holds that
\begin{equation}    \label{eq:CovUpdate}
    P_{k+1} = \Brac{ P_k^{-1} + \sigma^{-2} H_k\matT H_k }^{-1} + Q. 
\end{equation}
\end{lemma}
    
\begin{proof}
Notice under linearization assumption \cref{eq:Auto_Lin}, 
\[
    \bh_k(\btheta_k^{(i)}) - \mean{\bh_k(\btheta_k)} = H_k ( \btheta_k^{(i)} - \mean{\btheta}_k ). 
\]
Thus by definition,
\begin{align*}
    P_k^{\btheta,\by} :=~& \frac{1}{N-1} \sum_{i=1}^N \Brac{ \btheta_k^{(i)} - \mean{\btheta}_k } \otimes \Brac{ \bh_k(\btheta_k^{(i)}) - \mean{\bh_k(\btheta_k)} } = P_k H_k\matT, \\
    P_k^{\by,\by} :=~& \frac{1}{N-1} \sum_{i=1}^N \Brac{ \bh_k(\btheta_k^{(i)}) - \mean{\bh_k(\btheta_k)} } \otimes \Brac{ \bh_k(\btheta_k^{(i)}) - \mean{\bh_k(\btheta_k)} } = H_k P_k H_k\matT.
\end{align*}
The covariance update \cref{eq:ETKF_cov} can be written as 
\begin{align*}
    P_{k+1} =~& P_k - P_k^{\btheta,\by} \Brac{ P_k^{\by,\by}+ \sigma^2 I }^{-1} P_k^{\by,\btheta} + Q \\
    =~& P_k - P_k H_k\matT \Brac{ H_k P_k H_k\matT + \sigma^2  I  }^{-1} H_k P_k + Q \\
    =~& \Brac{ P_k^{-1} + \sigma^{-2} H_k\matT H_k }^{-1} + Q,
\end{align*}
where the last step uses the fact
\begin{align*}
    & \Brac{ P_k - P_k H_k\matT \Brac{ H_k P_k H_k\matT + \sigma^2  I  }^{-1} H_k P_k } \Brac{ P_k^{-1} + \sigma^{-2} H_k\matT H_k } \\
    =~& I + \sigma^{-2} P_k H_k\matT H_k - P_k H_k\matT \Brac{ H_k P_k H_k\matT + \sigma^2  I  }^{-1} \Brac{ H_k + \sigma^{-2} H_k P_k H_k\matT H_k } \\
    =~& I + \sigma^{-2} P_k H_k\matT H_k - P_k H_k\matT \Brac{ H_k P_k H_k\matT + \sigma^2  I  }^{-1} \Brac{ \sigma^2  I  +  H_k P_k H_k\matT } \sigma^{-2} H_k \\
    =~&  I + \sigma^{-2} P_k H_k\matT H_k - \sigma^{-2} P_k H_k\matT H_k = I. 
\end{align*}
This completes the proof.
\end{proof}    

\begin{lemma} \label{lem:W_bound}
Under assumption \cref{eq:Auto_Lin} and \cref{eq:Asm_Auto_Bound}, it holds that
\[
    W_k := P_k^{-1} - (I-K_k H_k)\matT P_{k+1}^{-1} (I-K_k H_k) \succeq (2 M^2 R^2 + \sigma^2 )^{-1} H_k\matT H_k. 
\]
\end{lemma}

\begin{proof}
    By \cref{eq:CovUpdate}, it holds that
    \begin{align}
        P_{k+1} = ( P_k^{-1} +& \sigma^{-2} H_k\matT H_k )^{-1} + Q \succeq \Brac{ P_k^{-1} + \sigma^{-2} H_k\matT H_k }^{-1}. \notag \\
        &\St~ P_{k+1}^{-1} \preceq P_k^{-1} + \sigma^{-2} H_k\matT H_k. \label{eq:Pf_PinvCtrl}
    \end{align}
    Therefore,
    \begin{equation}    \label{eq:Pf_W_Decomp}
    \begin{split}
        W_k =~ &P_k^{-1} - (I - K_k H_k)\matT P_{k+1}^{-1} (I - K_k H_k) \\
        \succeq~& P_k^{-1} - (I - K_k H_k)\matT \Brac{ P_k^{-1} + \sigma^{-2} H_k\matT H_k  } (I - K_k H_k) \\
        =~& H_k\matT K_k\matT P_k^{-1} + P_k^{-1} K_k H_k - H_k\matT K_k\matT P_k^{-1} K_k H_k\matT \\
        &- \sigma^{-2} (I - K_k H_k)\matT H_k\matT H_k (I - K_k H_k). 
    \end{split}
    \end{equation}
    Denote $\tilde{P}_k = H_k P_k H_k\matT$ for simplicity. Direct computation shows
    \begin{align*}
        H_k\matT K_k\matT P_k^{-1} =~& H_k\matT \Brac{ H_k P_k H_k\matT + \sigma^2  I  }^{-1} H_k \\
        =~&  H_k\matT \Brac{ \tilde{P}_k + \sigma^2  I  }^{-1} H_k  = P_k^{-1} K_k H_k,
    \end{align*}
    \begin{align*}
        H_k\matT K_k\matT P_k^{-1} K_k H_k\matT =~& H_k\matT \Brac{ H_k P_k H_k\matT + \sigma^2  I  }^{-1} H_k P_k H_k\matT \Brac{ H_k P_k H_k\matT + \sigma^2  I  }^{-1} H_k \\
        =~& H_k\matT \Brac{ \tilde{P}_k + \sigma^2  I  }^{-1} \tilde{P}_k \Brac{ \tilde{P}_k + \sigma^2  I  }^{-1} H_k. 
    \end{align*}
    \begin{align*}
        I - H_k K_k = I - H_k\matT P_k H_k &\Brac{ H_k P_k H_k\matT + \sigma^2  I  }^{-1} = \sigma^2 \Brac{ \tilde{P}_k + \sigma^2 I }^{-1}, \\
        \St (I - K_k H_k)\matT H_k\matT H_k (I - K_k H_k) =~& H_k\matT \Brac{ I - H_k K_k }\matT \Brac{ I - H_k K_k } H_k \\
        =~& \sigma^4 H_k\matT \Brac{  \tilde{P}_k + \sigma^2 I }^{-2} H_k.
    \end{align*}
    Plug into \cref{eq:Pf_W_Decomp}, we get
    \begin{align*}
        W_k \succeq~& H_k\matT \Big[ 2 ( \tilde{P}_k + \sigma^2 I)^{-1} - \Brac{ \tilde{P}_k + \sigma^2  I  }^{-1} \tilde{P}_k \Brac{ \tilde{P}_k + \sigma^2  I  }^{-1} \\
        &\qquad - \sigma^2 \Brac{  \tilde{P}_k + \sigma^2 I }^{-2} \Big]  H_k \\
        =~& H_k\matT ( \tilde{P}_k + \sigma^2 I)^{-1} H_k. 
    \end{align*}
    Since by assumption the parameters are bounded by $R$, so that
    \begin{align*}
        \norm{P_k} =~& \frac{1}{N-1} \norm{ \sum_{i=1}^N \btheta_k^{(i)} \otimes \btheta_k^{(i)} - N \mean{\btheta}_k \mean{\btheta}_k\matT  } \leq \frac{1}{N-1} \norm{ \sum_{i=1}^N \btheta_k^{(i)} \otimes \btheta_k^{(i)} } \\
        \leq~& \frac{1}{N-1}\sum_{i=1}^N  \normo{ \btheta_k^{(i)} \otimes \btheta_k^{(i)} } = \frac{1}{N-1}\sum_{i=1}^N  \normeo{ \btheta_k^{(i)}}^2 \leq \frac{N}{N-1} R^2 \leq 2R^2 . 
    \end{align*}
    \[
        \St ~ \normo{\tilde{P}_k} = \normo{ H_k P_k H_k\matT } \leq M^2 \normo{P_k} \leq 2 M^2 R^2.
    \]
    where we use $\norm{H_k} \leq M$. This implies that 
    \[
        \tilde{P}_k + \sigma^2 I \preceq (2 M^2 R^2 + \sigma^2 ) I ~\St~ W_k \succeq (2 M^2 R^2 + \sigma^2 )^{-1} H_k\matT H_k. 
    \]
    This completes the proof.
\end{proof}

\begin{lemma} \label{lem:KP_Ctrl}
    Under assumption \cref{eq:Auto_Lin}, \cref{eq:CovUpdate} and \cref{eq:Asm_Auto_Bound}, it holds that
    \[
        \norm{ K_k\matT P_{k+1}^{-1} } \leq \sigma^{-2} M.
    \]
\end{lemma}

\begin{proof}
    Notice $P_{k+1}^{-1} \preceq P_k^{-1} + \sigma^{-2} H_k\matT H_k$ (see \cref{eq:Pf_PinvCtrl}), so that
    \begin{align*}
        H_k P_k P_{k+1}^{-2} P_k H_k\matT \preceq ~& H_k P_k \Brac{ P_k^{-1} + \sigma^{-2} H_k\matT H_k}^2 P_k H_k\matT \\
        =~& \sigma^{-2} \Brac{ \sigma^2 I + H_k P_k H_k\matT } H_k \cdot \sigma^{-2} H_k\matT \Brac{ \sigma^2 I + H_k P_k H_k\matT }. 
    \end{align*}
    By definition, $K_k = P_k H_k\matT \Brac{ \sigma^2 I + H_k P_k H_k\matT }^{-1} $, so that
    \begin{align*}
        K_k\matT P_{k+1}^{-2} K_k &= \Brac{ \sigma^2 I + H_k P_k H_k\matT }^{-1} H_k P_k P_{k+1}^{-2} P_k H_k\matT \Brac{ \sigma^2 I + H_k P_k H_k\matT }^{-1} \\
        \preceq ~& \Brac{ \sigma^2 I + H_k P_k H_k\matT }^{-1} \cdot \sigma^{-4} \Brac{ \sigma^2 I + H_k P_k H_k\matT } H_k H_k\matT \Brac{ \sigma^2 I + H_k P_k H_k\matT } \\
        & \qquad \cdot \Brac{ \sigma^2 I + H_k P_k H_k\matT }^{-1} \\
        = ~& \sigma^{-4} H_k H_k\matT \preceq \sigma^{-4} M^2 I.
    \end{align*}
    So that $\norm{K_k\matT P_{k+1}^{-1}} = \norm{K_k\matT P_{k+1}^{-2} K_k}^{1/2} \leq \sigma^{-2} M$.
\end{proof}

\subsection{Proof for non-autonomous system}   \label{app:Pf_NonAuto_Converge}
The proof mostly repeats the previous section, and we only highlight the new steps to deal with the non-autonomous features. 

First notice that covariance update is unchanged, so that the estimations in \cref{lem:CovUpdate}, \cref{lem:W_bound} and \cref{lem:KP_Ctrl} are still valid for non-autonomous case.

\begin{proof}[Proof for \cref{thm:NonAuto_Converge}]
Under linearization \cref{eq:NonAuto_Lin}, it holds 
\[
    \mean{\btheta}_{k+1} = \mean{\btheta}_k - K_k H_k ( \mean{\btheta}_k - \btheta_k^\dagger ) + K_k \widehat{\beps}_k. 
\]
Notice $\btheta_{k+1}^\dag = \btheta_k^\dag + \delta \btheta_k$, we can rewrite the above equation as
\[
    \mean{\btheta}_{k+1} - \btheta_{k+1}^\dagger = ( I - K_k H_k )( \mean{\btheta}_k - \btheta_k^\dagger ) + K_k \widehat{\beps}_k - \delta \btheta_k. 
\]
Following the same strategy to derive \cref{eq:Pf_Decay}, one can show that 
\begin{align*}
    \normeo{ \mean{\btheta}_{k+1} - \btheta_{k+1}^\dagger }_{P_{k+1}}^2 \leq~& \normeo{ \mean{\btheta}_k - \btheta_k^\dagger }_{P_k}^2 - (\mean{\btheta}_k - \btheta_k^\dagger)\matT W_k ( \mean{\btheta}_k - \btheta_k^\dagger ) \\
    &+ 2 \Brac{ K_k \widehat{\beps}_k - \delta \btheta_k }\matT P_{k+1}^{-1} ( \mean{\btheta}_{k+1} - \btheta_k^\dagger ). 
\end{align*}
For the last term, notice 
\begin{align*}
    2 \Brac{ K_k \widehat{\beps}_k - \delta \btheta_k }\matT &P_{k+1}^{-1} ( \mean{\btheta}_{k+1} - \btheta_k^\dagger ) \leq 2 \Brac{ \normeo{ \widehat{\beps}_k } \normo{ K_k\matT P_{k+1}^{-1} } + \normeo{ \delta \btheta_k } \normo{P_{k+1}^{-1}} } \normeo{ \mean{\btheta}_{k+1} - \btheta_k^\dagger } \\
    \leq~& 4 R \sigma^{-2} M \normeo{ \widehat{\beps}_k } + 4 R \lambda_{\min}^{-1}(Q) \normeo{ \delta \btheta_k } .
\end{align*}
where we use $P_{k+1} \succeq Q \succeq \lambda_{\min}^{-1}(Q) I$ (see \cref{eq:CovUpdate}). 

Notice the rest procedures exactly repeats the previous steps, and we can get 
\[
    \frac{1}{T} \sum_{k=0}^{T-1} \normeo{ \mean{\by}_k - \bh_k(\btheta_k^\dagger)}^2 \leq \frac{C_1}{T} + \frac{C_2}{T} \sum_{k=0}^{T-1} \normeo{ \widehat{\beps}_k } + \frac{C_3}{T} \sum_{k=0}^{T-1} \normeo{ \delta \btheta_k },
\]
where $C_1,C_2$ are the same as before and $C_3 = 4 R \lambda_{\min}^{-1}(Q) ( 2 M^2 R^2 + \sigma^2 ) $.
\end{proof}

\end{document}